\documentclass[10pt, article, letter, twocolumn]{IEEEtran}

\usepackage{subfigure}
\usepackage{times}
\usepackage{latexsym}
\usepackage{morefloats}
\usepackage{amssymb}
\usepackage{amsmath}

\usepackage{enumitem}
\usepackage{mathrsfs}
\usepackage{amsgen,amsfonts,amsbsy,amsthm}
\usepackage{graphicx}
\usepackage{array}
\usepackage{hyperref}
\newcommand{\remark}[1]{{\emph{Remark:}} #1}

\newcommand{\BEQA}{\begin{eqnarray}}
\newcommand{\EEQA}{\end{eqnarray}}

\newcommand{\prob}{\mathsf{Pr}}

\newcommand{\gap}{\vspace{2mm}}
\newcommand{\comment}{\textbf{Comment: }}

\newcommand{\X}{\mathbf{X}}
\newcommand{\Y}{\mathbf{Y}}
\newcommand{\0}{\mathbf{0}}
\newtheorem{theorem}{Theorem}
\newtheorem{proposition}{Proposition}

\IEEEoverridecommandlockouts

\IEEEpubid{\makebox[\columnwidth]{} \hspace{\columnsep}\makebox[\columnwidth]{ }}

\begin{document}

\title{A Fast and Accurate Performance Analysis of Beaconless \mbox{IEEE 802.15.4} Multi-Hop Networks}
\author{Rachit Srivastava, Sanjay Motilal Ladwa, Abhijit Bhattacharya and Anurag Kumar\\Dept.\ of Electrical Communication Engineering\\Indian Institute of Science, Bangalore, 560012, India\\ email: \{rachitsri, sanjofpesit\}@gmail.com, \{abhijit, anurag\}@ece.iisc.ernet.in}


\maketitle
\begin{abstract}
We develop an approximate analytical technique for evaluating the performance of multi-hop networks based on beaconless \mbox{IEEE 802.15.4}, a popular standard for  wireless sensor networks. The network comprises sensor nodes, which generate measurement packets, relay nodes which only forward packets, and a data sink (base station). We consider a detailed stochastic process at each node, and analyse this process taking into account the interaction with neighbouring nodes via certain time averaged unknown variables (e.g., channel sensing rates, collision probabilities, etc.). By coupling the analyses at various nodes, we obtain fixed point equations that can be solved numerically to obtain the unknown variables, thereby yielding approximations of time average performance measures, such as packet discard probabilities and average queueing delays. The model incorporates packet generation at the sensor nodes and queues at the sensor nodes and relay nodes. We demonstrate the accuracy of our model by an extensive comparison with simulations. As an additional assessment of the accuracy of the model, we utilize it in an algorithm for sensor network design with quality-of-service (QoS) objectives, and show that designs obtained using our model actually satisfy the QoS constraints (as validated by simulating the networks), and the predictions are accurate to well within 10\% as compared to the simulation results.
\end{abstract}


\section{Introduction} \label{sec:introduction}
Wireless sensor networks (WSNs), a concept that originated in the mid-1990s, have now reached a stage in their evolution that we are beginning to see their actual deployment. It is not unreasonable to expect that in 10-15 years the world will be covered with wireless sensor networks with access to them via the Internet\cite{wireless-sensor-network}. \mbox{IEEE 802.15.4} \cite{IEEE802-15-4-06std} is a popular standard for the physical layer and medium access control for low-power wireless sensor networks. With the growing importance of wireless sensor networks in industrial applications \cite{gungor-hancke09industrial-WSN-survey}, we need analysis and design techniques for multi-hop IEEE~802.15.4 networks. The standard uses CSMA/CA for medium access control, and defines two types of CSMA/CA algorithms - beaconed (or slotted) and beaconless (or unslotted), respectively (see \cite{IEEE802-15-4-06std} for details). In this
paper we develop a new approximate analytical technique for multi-hop beaconless \mbox{IEEE 802.15.4} networks, and demonstrate the usefulness of such an analytical tool in designing multi-hop networks with Quality of Service (QoS) objectives.

\begin{figure}
\begin{center}
  \includegraphics[scale=0.31]{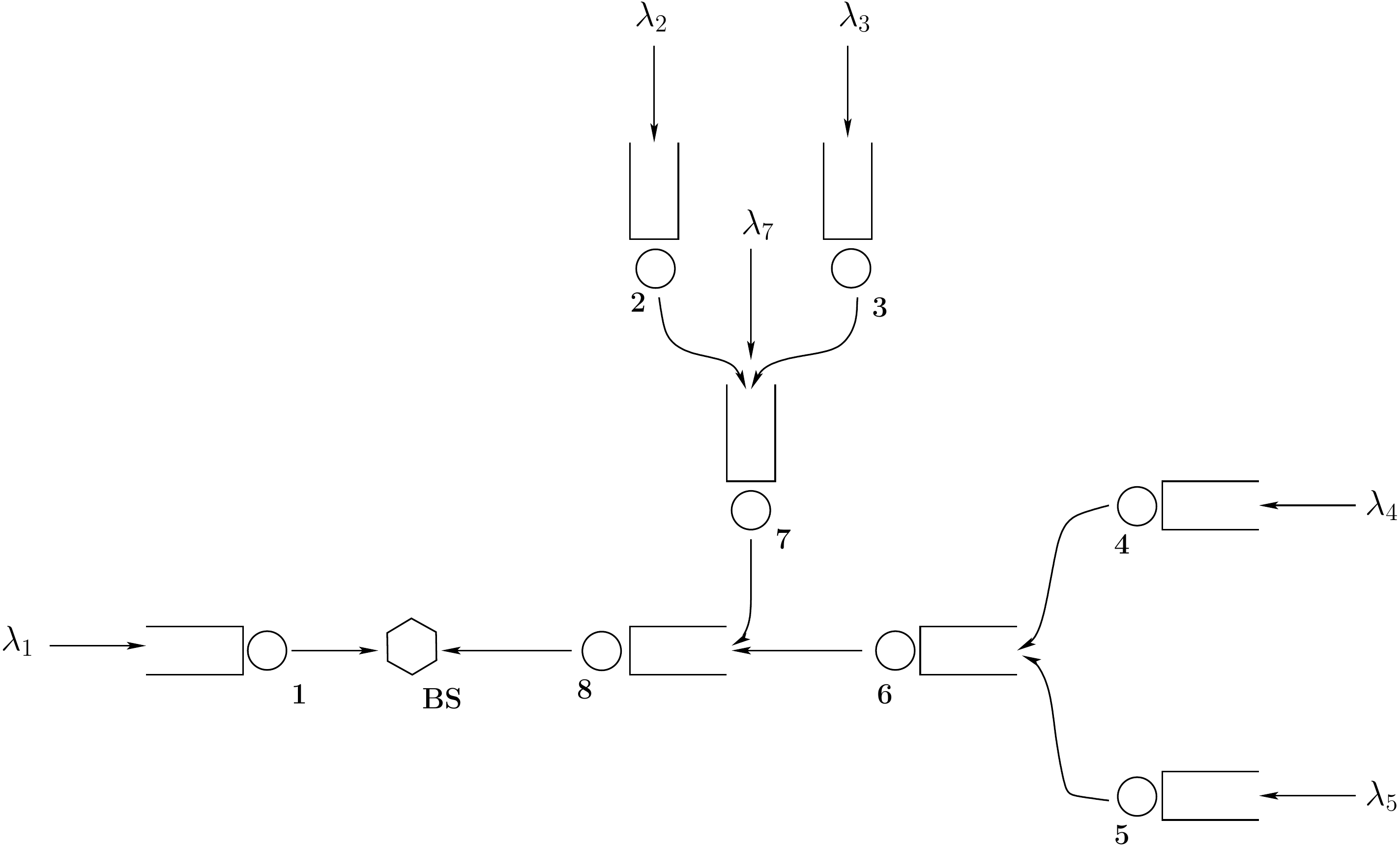}
  \caption{A WSN consisting of nodes arranged in a tree topology. Source node $i$ has packet generation rate $\lambda_i$. Relay nodes $6$, $7$ and $8$ forward the data of their predecessors to the base-station, BS.}
  \label{fig:network}
  \vspace{-7mm}
\end{center}
\end{figure}

Figure~\ref{fig:network} depicts the queueing schematic of a network of the type that we are concerned with in this paper. There are sensor nodes, namely, 1, 2, 3, 4, 5 and 7, that generate measurement packets according to independent point processes at rate $\lambda_i$ that need to be transported to the
base station (BS).  The sensor nodes can also serve as \emph{relays}
for the traffic of other sensor nodes, e.g., Node~7 is a relay for
nodes 2 and 3.  Two additional nodes, 6 and 8, serve only as relay
nodes, due to the limited range of the radios associated with the
sensor nodes. The nodes use unslotted CSMA/CA to contend for the wireless medium and transmit their packets. Even the stability analysis of such networks is a difficult problem (even for networks with single hop traffic; see \cite{bordenave} for a detailed discussion). In
principle, the entire network can be modeled via a large coupled
Markov chain, with the state at each node being the number of packets
in its queue, and the state of the contention process for the
head-of-the-line (HOL) packet. Such an approach is well known to be
intractable even for a network with a single contention domain (no hidden network) and saturated queues (see \cite{bianchi00performance,kumar-etal04new-insights}). Thus
all researchers have taken recourse to developing approximate
analyses. Table~\ref{tbl:related-literature} summarizes some of the recent work on modeling approaches for beaconless IEEE~802.15.4 networks where we have listed the major limitations of the proposed models as well. Owing to the need for low power operation and large coverage, in general, a sensor network is multi-hop with the presence of hidden nodes. Further, for networks that carry measurement traffic, the queue occupancies vary with time. Unlike other models, we consider all of these attributes in our modeling approach.

\scriptsize
\begin{table*}[p]
  \centering
  \caption{A comparison of the literature on analytical performance models of CSMA/CA based networks. The bold text in the "Network Scenario and Limitations" columns highlights the limitations of the model that we are able to handle in our model.}
  \begin{tabular}{|p{23mm}|p{45mm}|p{100mm}|}
   \hline
   \textbf{Authors}  &  \textbf{Network Scenario \& Limitations}  & \textbf{Modeling Approach}  \\
   \hline
   Kim et al. (2006)\cite{kim-etal06unslotted-CSMACA}
   	& 
	beaconless 802.15.4; \textbf{rare packet arrivals, star network; no hidden nodes; no queueing at nodes}; with ACKs
   	& 
	2-dim discrete time Markov chain for each node with states: (backoff stage, backoff counter) while in backoff, and separate states when transmitting and idling; with unknown transition probabilities; coupling the per node models to get the unknowns via fixed point approach
   \\ \hline
   Qiu et al. (2007)\cite{winet.qiu-etal07general-model-wireless-interference} 
   & 
   802.11 DCF; (un)saturated, \textbf{star network}; with hidden nodes and ACKs; infinite buffer space
   & 
   discrete time Markov chain with states as the set of transmitting nodes having unknown transition probabilities, together with RSS measurements between each pair of node to obtain CCA and packet failure probabilities; iterative procedure for solving unknowns
   \\ \hline
	Kim et al. (2008)\cite{kim-etal08modeling-802.15.4-single-hop}
    &
    beaconless 802.15.4; \textbf{star topology; no hidden nodes; no collisions}; infinite buffer space
    &
    unsaturated analysis of single hop network by simple mathematical model of M/G/1     queue; no collision is considered; computes QoS as packet delay and packet loss     probability; computes battery life time
   \\ \hline
Singh et al. (2008)\cite{winet.singh-etal08slotted-zigbee-star}
   &
   beaconed 802.15.4; (un)saturated, \textbf{star network; no hidden nodes}; with ACKs; infinite buffer; \textbf{equal traffic arrival rate at all nodes}
   &
   Markov renewal process at each node with cycle length dependent on the number of nodes available to attempt in the first backoff period of the cycle; transition probabilities and conditional expectations using node decoupling; coupling via stationary probabilites of remaining nodes; saturation analysis results used for the unsaturated case
   \\ \hline
    Buratti-Verdone (2009)\cite{buratti-verdone09modeling-802.15.4-beaconless}
    &
    beaconless 802.15.4; \textbf{star topology; no hidden nodes; no ACKs and retransmissions; no 
    queuing}
    &
   bidimensional process with backoff counter and backoff stage at time $t$; evaluation of sensing, transmission and success probability; computes mean energy spent by each node for a transmission 
   \\ \hline
   
   He et al. (2009)\cite{winet.he-etal09slotted-CSMACA} &  \textbf{beaconed} 802.15.4; \textbf{saturated, star network; no hidden nodes}; with ACKs;
   & 
   two Markov chains for each node: one embedded at the end of transmissions in the common channel, and the other during backoffs embedded at each slot; both having states as (backoff stage, backoff counter); chains coupled to get a fixed point approach
   \\ \hline
   Martalo et al. (2009)\cite{winet.martalo-etal09slotted-CSMACA}  
   &
   \textbf{beaconed} 802.15.4; unsaturated, \textbf{star network; no hidden nodes; no ACKs; finite buffer space}
   &
   two semi-Markov processes: one for tagged node having unknown transition probabilities, and the other for the shared radio channel that couples all the per node models to obtain the unknown quantities; MGFs used for finding queueing delays
   \\ \hline
   Goyal et al. (2009)\cite{goyal-etal09beaconless-zigbee} 
   & 
   beaconless 802.15.4; \textbf{unsaturated star network; no hidden nodes; no queueing at nodes}; with ACKs 
   &
   tracked the number of nodes with non-empty queues assuming that the probability of $m$ nodes having non-empty queues at any given time is same as the probability of $m-1$ empty nodes getting a new packet to send while a non-empty node is sending its current packet (for coupling individual node models); assumed that the CCA attempt processes of other nodes with non-empty queues are Poisson processes (decoupling)
   \\ \hline   
   Lauwens et al. (2009)\cite{lauwens-etal09unslotted-CSMACA-analysis} 
   & 
   beaconless 802.15.4; \textbf{saturated, star network; no hidden nodes; no ACKs}
   & 
   semi-Markov model for each node with unknown and non-homogeneous transition probabilities based on the backoff stage of the node; also found the distribution of backoff intervals and iterated over these together with the stationary probabilities of the semi-Markov processes
   \\ \hline
   Jindal-Psounis (2009)\cite{jindal09}
   &
   802.11 DCF; unsaturated, multi-hop network; with hidden nodes and ACKs; infinite buffer space; \textbf{modeling involves details that are specific to 802.11 such as the RTS/CTS mechanism};   
   &
   For each of four possible classes of two-edge topologies, derived collision and idle probabilities for each edge in terms of expected service times of the edges; for a general topology, decomposed the local network topology around each edge into a number of two-edge topologies, and derived the collision and idle probabilities for the edge in terms of those of the two-edge topologies in the decomposition; using these probabilities, set up and solved a Markov chain model for each edge (with states as the current backoff window, backoff counter, and time since the last successful/unsuccessful RTS/CTS exchange) to express the expected service times of the edges via a set of fixed point equations, which were solved iteratively to find the achievable rate region under 802.11 in the given topology
  \\ \hline
   DiMarco et al. (2010)\cite{dimarco-etal10modeling-ieee-802-15-4-multihop} 
   & 
   beaconless 802.15.4; unsaturated, multi-hop network; \textbf{no queueing at nodes}; with hidden nodes and ACKs 
   & 
   3-dim discrete time Markov chain for each node having state: (backoff stage, backoff counter, retransmission counter) with CCA and collision probabilities (unknown) as state transition probabilities (decoupling); assumed independence of node processes to find expressions for the CCA and collision probabilities using stationary probabilities of different Markov chains (coupling) 
   \\ \hline
   Sen-De (2010)\cite{sen-de10modelling-802.15.4-beaconless}
   &
   beaconless 802.15.4; unsaturated, multi-hop network; with hidden nodes; \textbf{no queuing; analysis does not match simulation well}
   &
   1-dim Markov model for transmitting node with states being in CCA, backoffs states and data transmission; steady state transmission probability and probability of success (throughput) are obtained through fixed point iterations
   \\ \hline
   Shyam-Kumar (2010)\cite{shyam-kumar10me-thesis} 
   & 
   beaconless 802.15.4; (un)saturated, multi-hop network; with hidden nodes and ACKs; infinite buffer; \textbf{analysis does not match simulation well} 
   & 
   saturated analysis using continuous time Markov chain with states as the set of transmitting nodes having unknown transition probabilities; found CCA and failure probabilities using steady state probability of the Markov chain; unsaturated network modeled as a Markov renewal process constituted by the nodes having non-empty queues, saturation analysis results used in each cycle with the set of non-empty nodes
   \\ \hline
  Marbach et al. (2011)\cite{marbach11}
  &
  \textbf{non-adaptive CSMA characterized by a \emph{fixed} attempt probability on each link}, and a non-zero sensing period on each link; asynchronous, unsaturated, multihop network; \textbf{over-simplified model - does not conform to any existing standard}
  &
  intuitively formulated fixed point equations involving node idling probabilities, and transmission attempt rates; showed uniqueness of the fixed point, and asymptotic accuracy for large networks with small sensing period and appropriately decreasing link attempt probabilities; for a given sensing period, characterized the achievable rate region of the simplified CSMA policy by defining a set of arrival rate vectors, and showing that for every member in that set, there exists a policy (a vector of attempt probabilities) such that the link service rates obtained from the fixed point equations exceed the link arrival rates\\
\hline
  \end{tabular}
  \label{tbl:related-literature}
  \vspace{-5mm}
\end{table*}
\normalsize

\subsection*{Our Contributions and Comparison with Related Work}  

We consider a multi-hop WSN consisting of static sources and relays arranged in a tree topology operating in the beacon-less mode with acknowledgements (ACKs). Each node has an infinite buffer space and may operate in the saturated or unsaturated regime with fixed packet length. Different analysis techniques are developed for networks with hidden nodes and networks without hidden nodes due to the difference in activity lengths perceived by a node in the presence and absence of hidden nodes. Under certain approximations, we model the stochastic process evolving at a node by incorporating the influence of the other nodes in the network by their (unknown) time averaged statistics, and then couple these individual node processes via a system of fixed point equations, which is solved using an iterative scheme to obtain the unknown variables. Although this \emph{decoupling} (or \emph{mean-field}) approximation is popular in such situations, our more detailed model incorporating several issues not considered together hitherto requires a careful handling of the analysis. We identify and calculate two QoS measures for each source node, viz., the packet delivery probability and end-to-end packet delay, in terms of these variables. We observe that in a multi-hop IEEE~802.15.4 network, the packet delivery probability falls sharply before the end-to-end packet delay becomes substantial.
 
The rest of the paper is organized as follows: In Section~\ref{sec:overview-of-standard}, we give an overview of the
unslotted CSMA/CA mechanism. Section~\ref{sec:analytical-model}
explains the node behaviour and elaborates upon the different
analysis techniques used for networks without hidden nodes and
networks with hidden nodes. Section~\ref{sec:numerical-and-simulation-results} compares our analysis with simulations. Section~\ref{sec:network-design} demonstrates how the analysis can be used for designing multi-hop networks with QoS objectives. Finally, we conclude the paper in
Section~\ref{sec:conclusion}.


\section{Unslotted CSMA/CA for Beacon-less IEEE 802.15.4 Networks}\label{sec:overview-of-standard}

A typical node behaviour under unslotted CSMA/CA is shown in Figure~\ref{fig:evolution}. A node with an empty queue remains idle until it generates a packet or receives one from its predecessor nodes. When a node has data to send (i.e., has a non-empty queue), it initiates a random back-off with the first back-off period being sampled uniformly from $0$ to \emph{${2^{macminBE}-1}$}, where \emph{macminBE} is a parameter fixed by the standard. For each node, the back-off period is specified in terms of slots where a slot equals 20 symbol times ($T_{\mathsf{s}}$), and a symbol time equals 16 $\mu s$.\footnote{Note, however, that there is no central coordinator that maintains synchrony of the slots \emph{across the nodes}, and hence the CSMA/CA protocol is unslotted.} The node then performs a CCA (\emph{Clear Channel Assessment}) to determine whether the channel is idle. If the CCA succeeds, the node does a \emph{Rx-to-Tx turnaround}, which is $12$ symbol times, and starts transmitting on the channel. The failure of the CCA starts a new back-off process with the back-off exponent raised by one, i.e., to \emph{macminBE+1}, provided that this is less than its maximum value, \emph{macmaxBE}. The maximum number of successive CCA failures for the same packet is governed by \emph{macMaxCSMABackoffs}, exceeding which the packet is discarded at the MAC layer. The standard allows the inclusion of acknowledgements (ACKs) which are sent by the intended receivers on a successful packet reception. Once the packet is received, the receiver performs a \emph{Rx-to-Tx turnaround}, which is again $12$ symbol times, and sends a $22$ symbol fixed size ACK packet. A successful transmission is followed by an \emph{InterFrame Spacing}(IFS) before sending another packet.

When a transmitted packet collides or is corrupted by the PHY layer noise, the ACK packet is not generated, which is interpreted by the transmitter as failure in delivery. The node retransmits the same packet for a maximum of \emph{aMaxFrameRetries} times before discarding it at the MAC layer. After transmitting a packet, the node turns to the Rx-mode and waits for the ACK.  The \emph{macAckWaitDuration} determines the maximum amount of time a node must wait for in order to receive the ACK before concluding that the packet (or the ACK) has collided. The default values of \emph{macminBE}, \emph{macmaxBE}, \emph{macMaxCSMABackoffs}, and \emph{aMaxFrameRetries} are 3, 5, 4, and 3 respectively.


\section{Description of the Modeling Approach} \label{sec:analytical-model}

As shown in Figure~\ref{fig:network}, the network consists of sensors, relays and a base-station (BS). We consider time invariant links and static (immobile) nodes. The links are characterized by a predefined target packet error rate (PER) (e.g., a PER of 1\% on each link). We assume that the sensors generate traffic according to independent point processes. These constitute the aggregate external arrival process for the sensor network. Each node transmits its data to the next-hop node according to the topology; throughout this paper, \emph{we shall work with a tree topology so that each transmitter has exactly one receiver node}. The intermediate nodes along a route may be relays in which case they simply forward the incoming traffic, or they may be sensors which transmit their own packets as well as the received packets. Based on the network congestion, the nodes may discard packets due to consecutive failed CCAs or frame retries. Figure~\ref{fig:evolution} shows a typical sample path of the process evolution at a node. 
\begin{figure}[ht]
	\begin{center}
		\includegraphics[height=4cm, width=9cm]{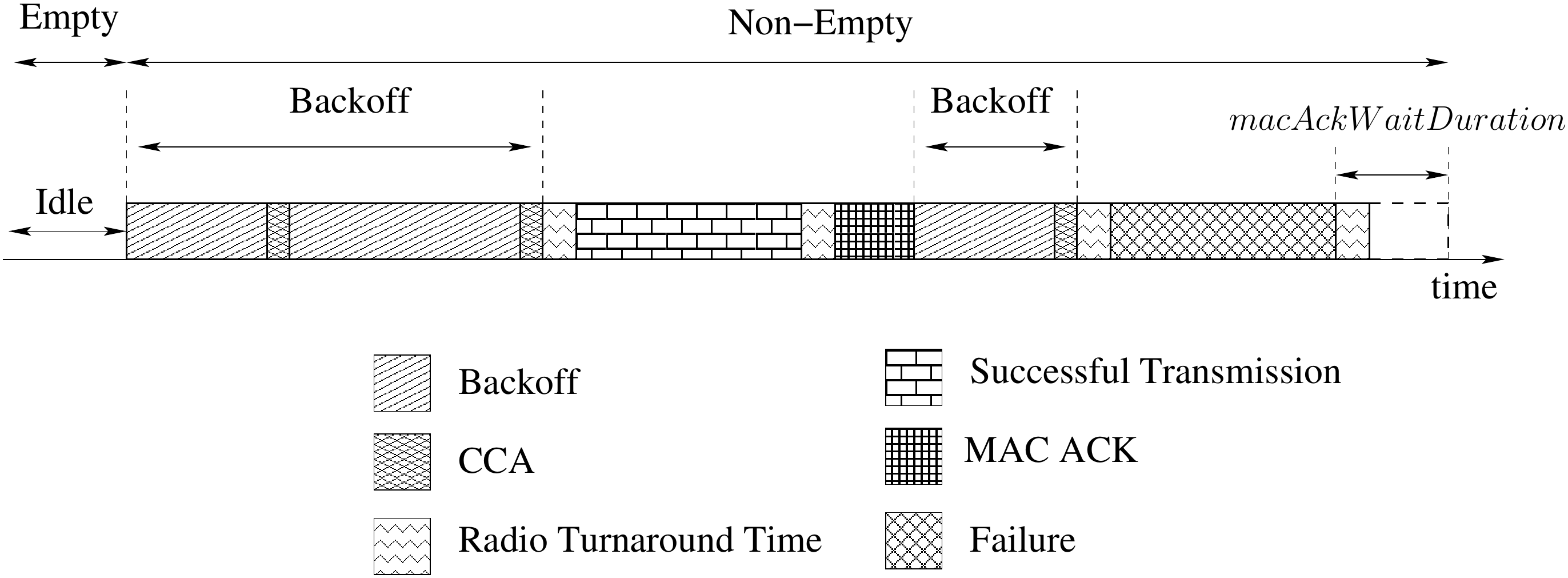}
		\caption{Process evolving at a tagged node as it goes through various states.}
		\label{fig:evolution}
		\vspace{-5mm}
	\end{center}
\end{figure}

\subsection{Modeling Simplifications}
\label{subsec:model-simplify}
The process evolutions at various nodes in the system are coupled, and each node (source or relay) can be present in any state; further, different nodes can interact with different subsets of nodes based on their Carrier Sense (CS) range and positions.\footnote{A node $j$ is said to be within the Carrier Sense (CS) range of a node $i$ if any channel activity (transmission) due to node $j$ can be detected (but not necessarily decoded) by node $i$.}  Hence, the random process representing the system is multi-dimensional and non-homogeneous. The exact analysis of such a process is intractable. We, therefore, make several simplifications:
\begin{description}[leftmargin=0cm]
\item [(S1)]All packets are of the same fixed length. \emph{Hence, the DATA transmission duration is fixed.}
\item [(S2)]The internode propagation delays are of the order of nanoseconds, and hence can be taken to be zero.
\item [(S3)]Collision model (zero packet capture): If a receiver is in the communication range of two or more nodes that are transmitting simultaneously, it does not receive any of those transmitted packets, i.e, all the packets involved in a collision get corrupted unlike an SINR model, where some of the interfering packets may still be received successfully, depending on the SINR threshold. 
\item [(S4)]The Inter Frame Spacing (IFS) is neglected. \emph{It means that the receiver does not discard the packet received during the IFS.}
\item [(S5)] By (S3), a packet is necessarily bad if any other packet is being heard by the same receiver. Even if this is not the case, a packet can be in error due to noise. The packet error probability on each link is fixed and known. \emph{As already mentioned, the links are not time varying and we can empirically obtain the error characteristics of a link.}
\item [(S6)]ACK packets are short and, therefore, not corrupted by PHY layer noise. 
\item[(S7)]A node's CCA succeeds when there is no transmission by any node in its Carrier Sense (CS) range at the time of initiation of its CCA. \emph{Recall that in the standard, the channel state is averaged over the $8$~symbol duration.}
\item [(S8)]The time taken by a transmitting node for the activities of successful transmission and collision are the same. We denote this time by $T_{tx}$. \emph{If the transmitted data collides at the receiver, the macAckWaitDuration for the transmitting node is equal to the sum of the turnaround time of $12~T_s$ and the ACK duration of $22~T_s$, a total of $34~T_s$ (see Figure~\ref{fig:transmission}).}
\item [(S9)]We assume symmetry in carrier sensing, and signal reception, i.e., \emph{if a node $i$ can detect (respectively decode) the packet transmissions of another node $j$, then node $j$ can also detect (resp. decode) the packet transmissions of node $i$.}
\end{description}

\begin{figure}[htbp]
	\begin{center}
		\includegraphics[scale=0.45]{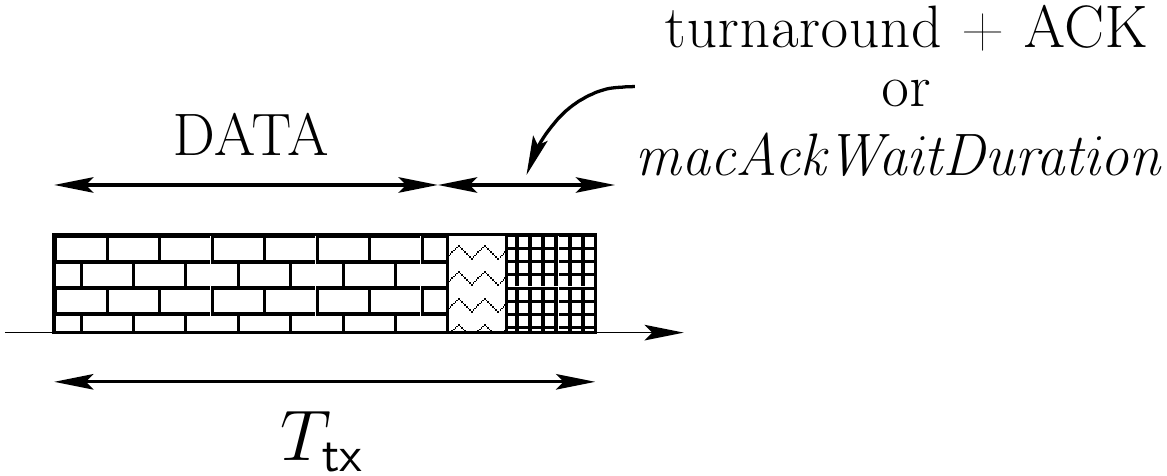}
		\caption{The transmission period $T_{\mathsf{tx}}$ which includes the DATA transmission time $T_x$, the turnaround time of $12~T_s$, and the MAC ACK duration of $22~T_s$ (or the \emph{macAckWaitDuration} of $34~T_s$, by \textbf{(S8)}).}	
		\label{fig:transmission}
		\vspace{-5mm}
	\end{center}
\end{figure}

\gap
\noindent
\emph{Decoupling approximation:} We employ an approximation whereby we model each node separately, incorporating the influence of the other nodes in the network by their average statistics, and as if these nodes were independent of the tagged node. This commonly used approximation is also called a ``mean field approximation'', and has been widely used in the wireless networking literature \cite{bianchi00performance,kumar-etal04new-insights,jindal09}. 

\subsection{Modeling Node Activity}
\label{subsec:node-activity-modeling}
The channel activity perceived by any node $i$ is only due to the nodes in its CS set (i.e., the set of nodes within its CS range) denoted by $\Omega_i$. Note that $\Omega_i$ does not include node $i$. Let us consider the durations during which the queue at node $i$ is non-empty. During these times, node $i$ alternates between performing CCAs and transmitting the packets from its queue. We model the CCA attempt process at node $i$ as a Poisson process of rate $\beta_i$ \emph{conditioned on being in backoff periods}\cite{goyal-etal09beaconless-zigbee}. For each node $j\in \Omega_i$, consider only those times at which the node is not transmitting (i.e., its queue is empty, or the node is in backoff). We model the CCA attempt process of each node $j\in \Omega_i$ conditioned on these times by an independent Poisson process with rate $\overline{\tau}_j^{(i)}$, $j\in\Omega_i$. By modeling simplification \textbf{(S1)}, we assume that all packets entering node $i$ have the same fixed length, and hence take the same amount of time when being transmitted over the medium (denoted by $T_{\mathsf{tx}}$). 

Let us now remove the time intervals at node $i$ during which the queue at this node is empty, thereby concatenating all the busy periods at the node. In this conditional time, as a result of the Poisson point process assumption for the attempt process at node $i$, and also for the attempt processes of the nodes in $\Omega_i$, we observe that instants at which packet attempts complete (equivalently, new backoff intervals for the attempts start) are renewal instants. These are denoted by $X^{(k)}_i$ in Figure~\ref{fig:ARP_calc_of_rates}. The corresponding renewal lifetimes are denoted by $W_i^{(k)}$. 

Figure~\ref{fig:system} shows details of a renewal cycle of node $i$. The dashed transmission durations belong to nodes other than $i$. The cycle always ends with a transmission from node $i$. This last transmission in a cycle always corresponds to a non-collided packet, which, however, could be received in error. 
\begin{figure}[h]
\begin{center}
\includegraphics[width=9cm]{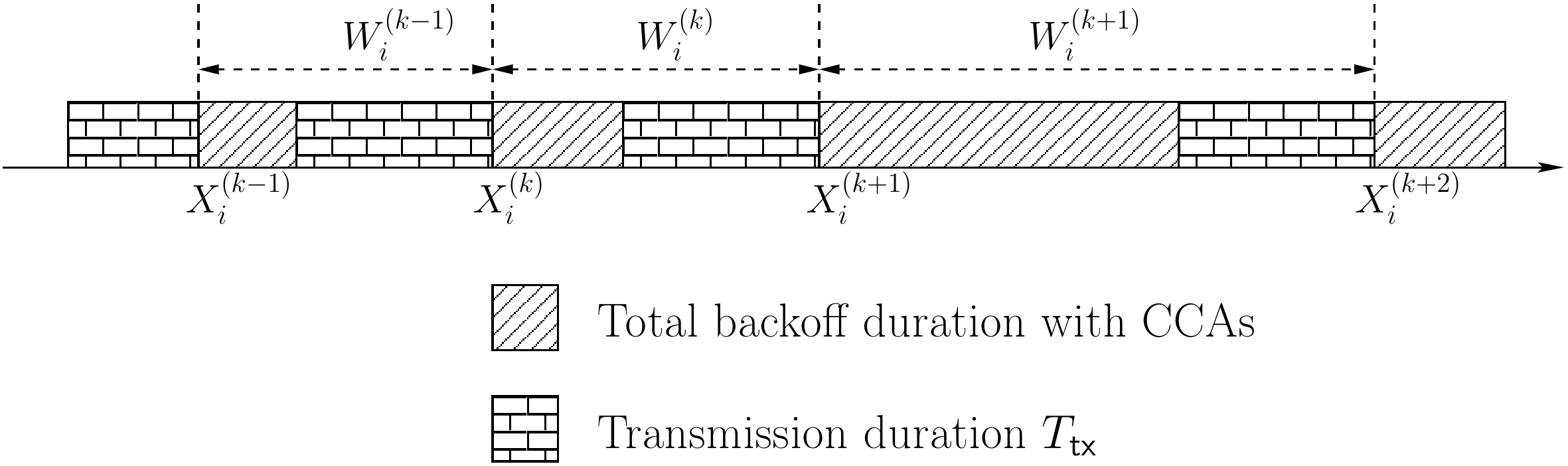}
\caption{The renewal process obtained by observing the process at a node $i$ after removing all the idle time periods from the original process shown in Figure~\ref{fig:evolution}. The renewal epochs are denoted by $\{X_i^{(k)}\}$ and the cycle lengths by $\{W_i^{(k)}\}$.}
\label{fig:ARP_calc_of_rates}
\vspace{-5mm}
\end{center}
\end{figure}
\begin{figure}[h]
	\begin{center}
		\includegraphics[scale=0.4]{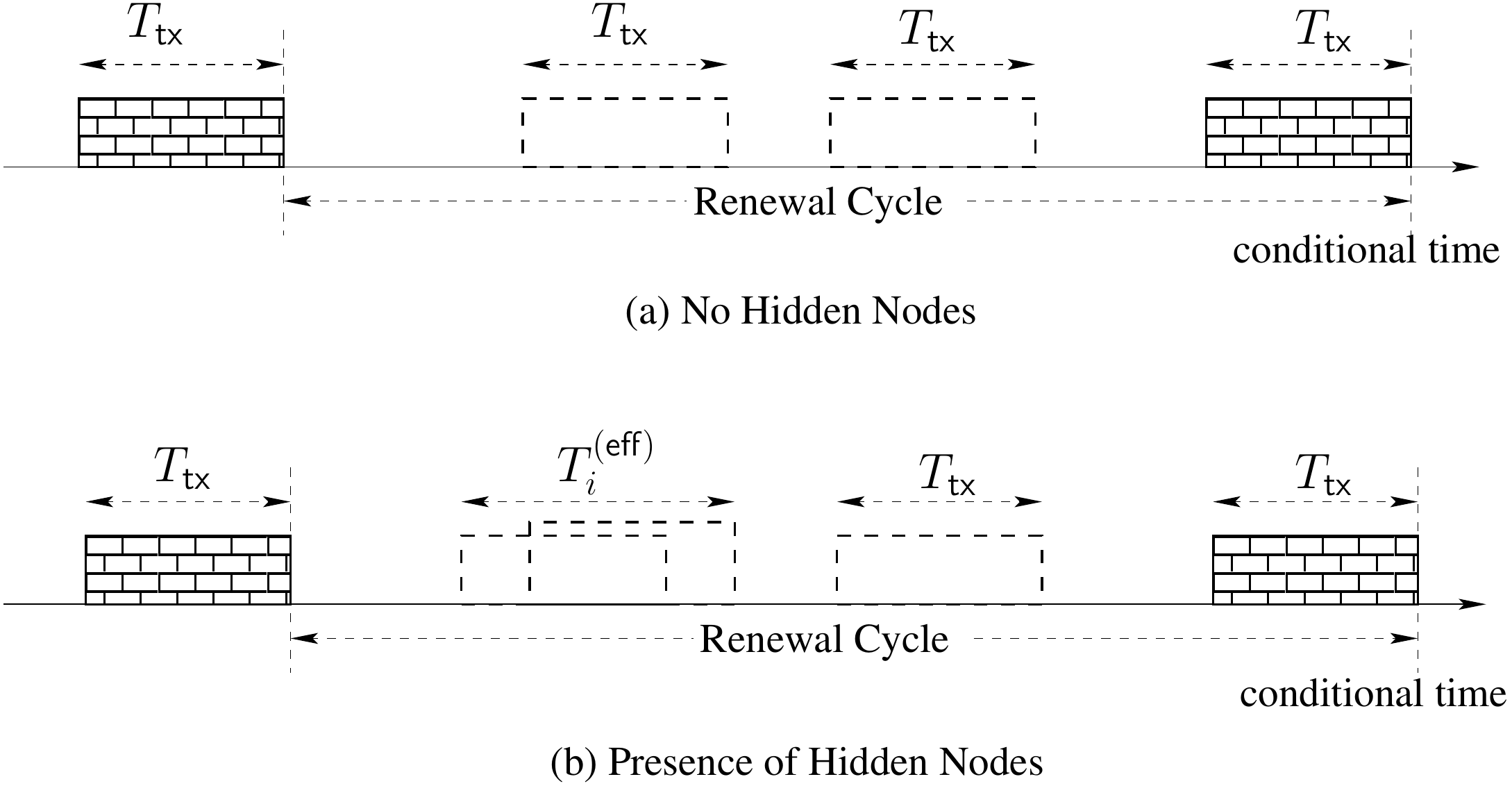}
		\caption{Renewal cycles at a tagged node, say $i$: (a) a renewal cycle when there are no hidden nodes in the system (all activity periods on the channel are of length $T_{\mathsf{tx}}$); (b) a renewal cycle when there are hidden nodes (due to hidden node collisions, the activity periods on the channel can be longer than $T_{\mathsf{tx}}$.)}
		\label{fig:system}
		\vspace{-5mm}
	\end{center}
\end{figure}
Before discussing the difference in activity lengths in the two network types (as shown in Figure~\ref{fig:system}), we list the various packet collision scenarios. Packet collisions in a wireless network with a CSMA MAC can occur due to the presence of hidden nodes. In addition, \emph{there are scenarios where transmissions from two transmitters placed within the CS range of each other can overlap, known as \textbf{Simultaneous Channel Sensing} and \textbf{Vulnerable Window}} (also known as First and Second Collision window respectively, in \cite{goyal-etal09beaconless-zigbee}). These two collision scenarios are explained in Figures~\ref{fig:simul_channel} and \ref{fig:CCAfail}, respectively. Note that we are making use of \textbf{(S2)} to avoid including insignificant timing details here.
\begin{figure}[h]
\begin{center}
	\includegraphics[scale=0.45]{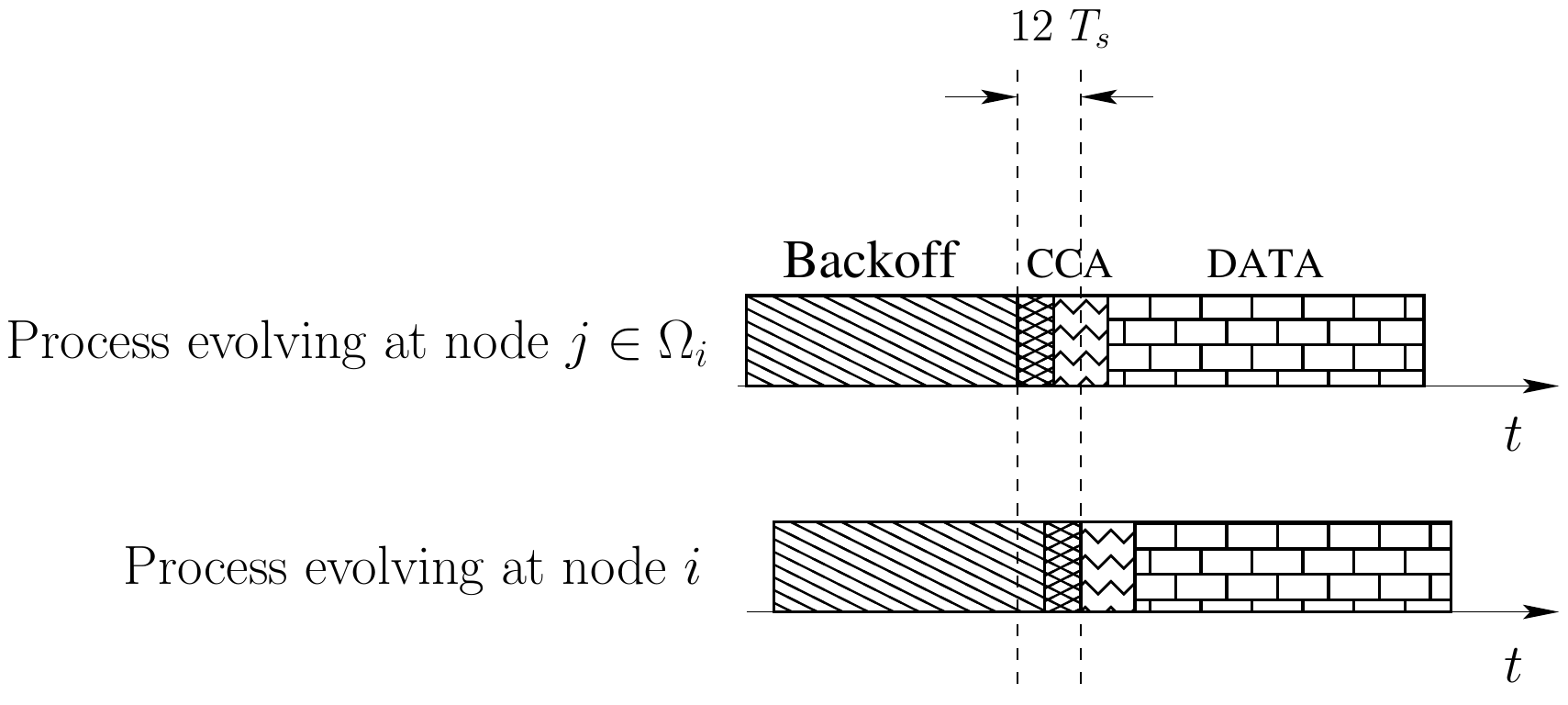}
	\caption{Node $j\in \Omega_i$ finishes its backoff, performs a CCA, finds the channel idle and starts transmitting the DATA packet. Node $i$ finishes its backoff anywhere in the shown $12~T_s$ duration and there is no other ongoing transmission in $\Omega_i$, its CCA succeeds and it enters the transmission duration. As a result, the DATA packets may collide at $r(i)\in \Omega_j$ and/or $r(j) \in \Omega_i$.}
	\label{fig:simul_channel}
	\vspace{-5mm}
\end{center}
\end{figure}
\begin{figure}[h]
\begin{center}
\includegraphics[scale=0.45]{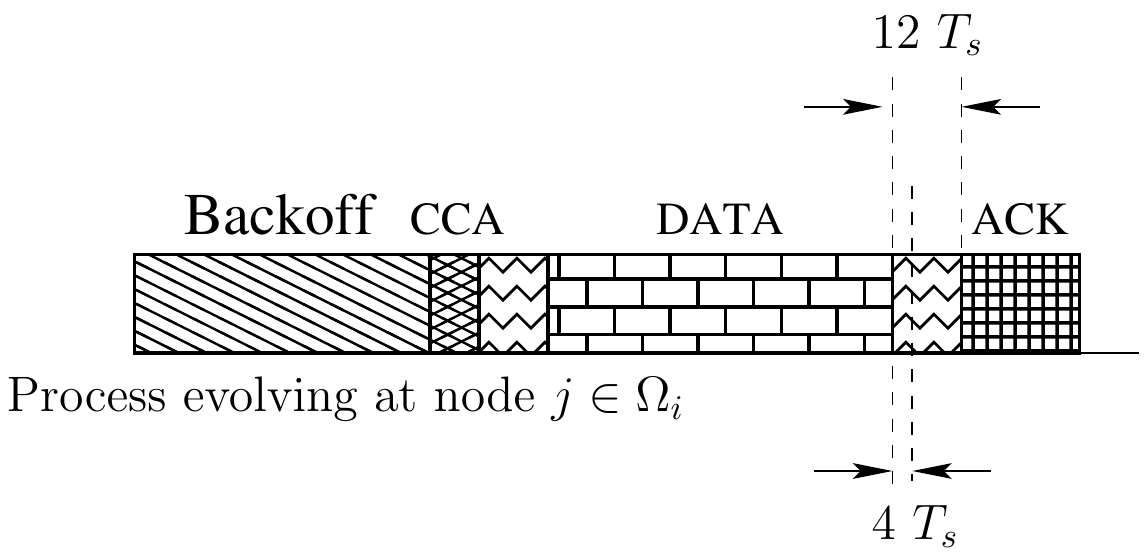}
	\caption{Vulnerable window of $4~T_s$ in the transmission period of node $j\in \Omega_i$ during which node $i$'s CCA attempt would be successful.}
	\label{fig:CCAfail}
	\vspace{-5mm}
\end{center}
\end{figure}

Since the vulnerable window size is small compared to the data transmission duration of a node (see Figure~\ref{fig:CCAfail}), \emph{we shall assume that the probability of a node's CCA initiating in the vulnerable window of another node is negligible}, and therefore, neglect the resulting packet collisions. Further, since the ACK packet size is just a small fraction of DATA packet size (e.g. compare an ACK packet ($22$~symbols) with a DATA packet of length $260$~symbols at PHY layer), \emph{we shall assume that the probability of packet collision involving ACK packets is negligible}.  

We now return to elucidate the difference in length of activity periods in a renewal cycle.
\begin{description}[leftmargin=0cm]
\item[(a) Absence of Hidden Nodes:] When a node transmits, all the other nodes in the network can hear it, resulting in CCA failures for other nodes that try to assess the channel in the transmission period. Also if two nodes are involved in simultaneous channel sensing, the activity period may extend from $T_{\mathsf{tx}}$ to a maximum of $(T_{\mathsf{tx}}+12~T_s)$. Since $12~T_s<<T_{\mathsf{tx}}$, we can assume that the activity period is only a single transmission period of duration $T_{\mathsf{tx}}$. 
\item[(b) Presence of Hidden Nodes:] Since a node's hearing capacity is limited, it may not perceive the activities of all the nodes in the network which may cause dilation of activity period.
\begin{figure}[ht]
\begin{center}
		\includegraphics[scale=0.3]{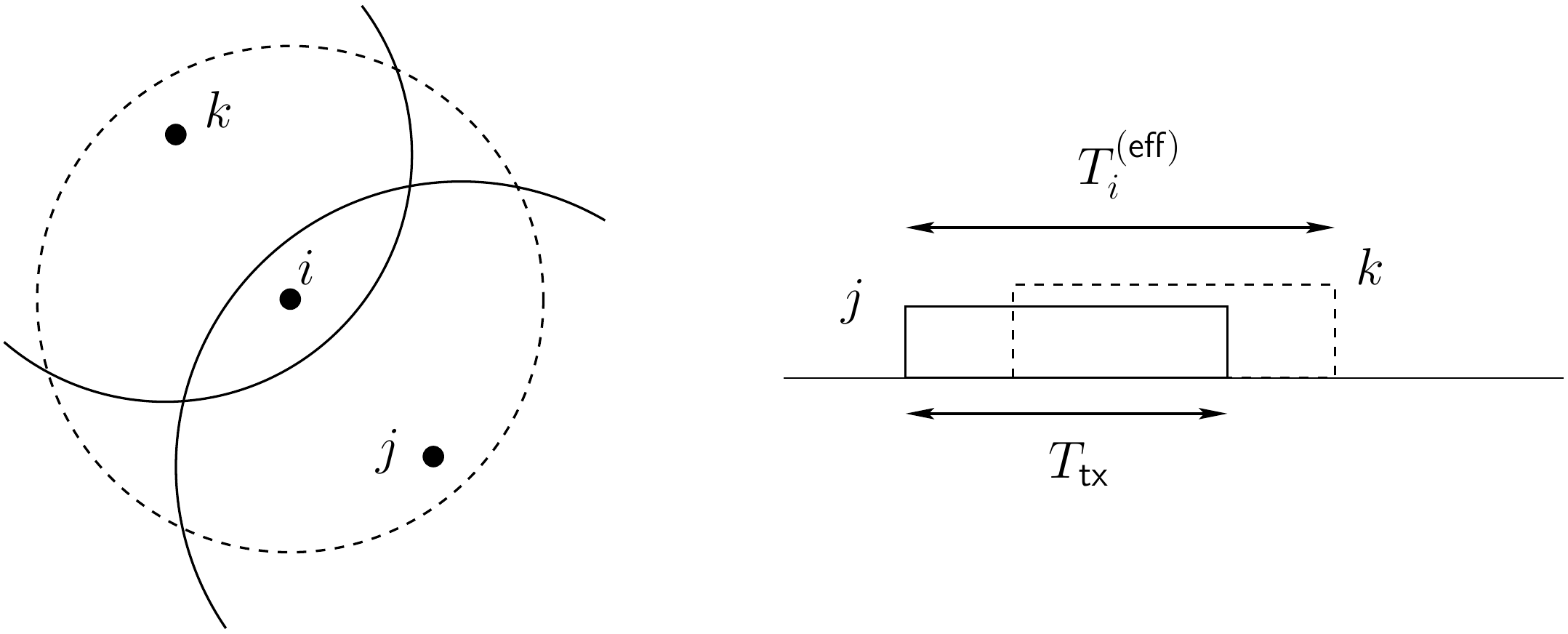}
		\caption{\emph{Dilation of transmission period as perceived by node $i$}.}
		\label{fig:t_eff}
		\vspace{-5mm}
\end{center}
\end{figure}
As shown in the Figure \ref{fig:t_eff}, node $j$ and $k$ are hidden from each other with respect to the receiver node $i$. Suppose node $j$ starts transmitting while nodes $i$ and $k$ are silent. Since node $k$ cannot hear node $j$, it can start its own transmission while the transmission from node $j$ is still going on. For node $i$, this looks like a dilated transmission period, whose length is denoted by $T_{i}^{(\mathsf{eff})}$. 
\end{description}

\section{Derivation of Node Interaction Equations}
\label{sec:fixed-pt-eqns-derivation}
In the following subsections, we consider the node processes described in Section~\ref{sec:analytical-model}, and identify certain useful steady state quantities at each node, e.g., CCA attempt rate. Then we employ the decoupling approximation (described in Section~\ref{sec:analytical-model}), and detailed stochastic analysis, to write down certain fixed point equations relating these quantities. In doing this, \emph{we assume the network to be stable so that these steady state quantities exist}. Then in Section~\ref{sec:iterative-scheme}, we use an iterative scheme to solve for the unknown variables, and finally use the quantities thus obtained to evaluate the network performance, namely, end-to-end delay and packet delivery probability from each source node. The analysis will be presented for the more general case where hidden nodes may be present in the system; wherever necessary, we shall mention the changes required in the expressions for the special case where no hidden nodes are present in the system.

\subsection{Derivation of Fixed Point Equations}
\label{subsec:fixed-pt-derivation}
We begin by analyzing the process at node $i$ shown in Figures~\ref{fig:ARP_calc_of_rates} and \ref{fig:system}. Consider an epoch at which node $i$ starts its backoffs in order to send the first packet in its queue. At this instant, each node $j \in \Omega_i$ is either empty or performing backoffs for its HOL packet or transmitting its HOL packet. Recall that the rate at which node $i$ completes backoffs during its backoff periods is denoted by $\beta_i$. The activity of node $i$ is affected by the residual rate of CCA attempts by $j \in \Omega_i$ after removing those CCA attempts that fail due to nodes that are hidden from node $i$; this rate is denoted by $\overline{\tau}^{(i)}_j$. Modeling these conditional attempt processes as independent Poisson processes, we  define:
\begin{eqnarray*}
\eta_i = \frac{\beta_i}{\beta_i + \displaystyle{\sum_{j\in\Omega_i}\overline{\tau}_j^{(i)}}} \quad &;& \quad  g_i = \frac{1}{\beta_i + \displaystyle{\sum_{j\in\Omega_i}\overline{\tau}_j^{(i)}}} \\
\text{and} \qquad c_i &=& \Bigg(1-e^{-12~T_\mathsf{s}\beta_i}\Bigg) 
\end{eqnarray*}
Thus, $\eta_i$ is the probability that node $i$ makes a CCA attempt before any node in $\Omega_i$, $g_i$ is the mean time until the first CCA attempt, and $c_i$ is the probability that after an attempt by a node in $\Omega_i$, node $i$ attempts within the 12 symbol vulnerable period, thus causing a ``simultaneous'' channel sensing collision (see Figure~\ref{fig:simul_channel}).

\subsubsection{Perceived CCA attempt rate, $\overline{\tau}_j^{(i)}$}
\noindent
In a network with hidden nodes, in general, $\Omega_j \setminus \Omega_i$ is nonempty. Hence, CCA attempts of  node $j$ could be blocked by the activities of nodes in $\Omega_j \setminus \Omega_i$, thus leading to the node~$i$ perceiving a CCA attempt rate from node~$j$ that is less than the CCA attempt rate when $\Omega_j \setminus \Omega_i = \phi$ (absence of hidden nodes). We have denoted this effective CCA attempt rate of node $j$ as perceived by node~$i$ by $\overline{\tau}_j^{(i)}$ for all $j \in  \Omega_i$. Define
\begin{description}
\item [$N_j^{(\mathsf{cca})}(t)$:] Total number of CCAs attempts of Node $j$ in time $(0, t]$ 
\item [$N_j^{(\mathsf{ccaf_{-i}})}(t)$:] Number of failed CCA attempts of Node $j$ due to the nodes in $\Omega_j \setminus \Omega_i$ in time $(0, t]$
\item [$I_j(t)$:] The time during which node $j$ is not transmitting in time $(0, t]$
\item [$T_{j}^{\mathsf{(bo)}}(t)$:] Time for which node $j$ is in backoff in the interval $(0, t]$
\item [$T_{j}^{\mathsf{(ne)}}(t)$:] Queue non-empty period of node $j$ in the interval $(0, t]$
\end{description}

Note that with these definitions, we have 
\begin{align*}
\beta_j &= \lim_{t\rightarrow \infty}\frac{N_{j}^{(\mathsf{cca})}(t)}{T_{j}^{\mathsf{(bo)}}(t)}\text{ a.s}
\end{align*}

We can further define,

\begin{align*}
b_j &= \lim_{t\rightarrow \infty}\frac{T_{j}^{\mathsf{(bo)}}(t)}{T_{j}^{\mathsf{(ne)}}(t)}\\
q_j &= \lim_{t\rightarrow \infty}\frac{T_{j}^{\mathsf{(ne)}}(t)}{t}
\end{align*}
where $b_j$ can be interpreted as the long-term fraction of time node $j$ is in backoff provided it is non-empty, and $q_j$ is the long-term fraction of time that node $j$ is non-empty, which is also the queue non-empty probability of node $j$, assuming the system is ergodic.

Then, we can write
\begin{align*}
  \overline{\tau}_j^{(i)} &= \lim_{t\rightarrow \infty}\frac{N_j^{(\mathsf{cca})}(t) - N_j^{(\mathsf{ccaf_{-i}})}(t)}{I_j(t)} \text{ a.s}\\ 
                          &= \lim_{t\rightarrow \infty} \frac{\frac{N_{j}^{(\mathsf{cca})}(t)}{T_{j}^{\mathsf{(bo)}}(t)}\times \frac{T_{j}^{\mathsf{(bo)}}(t)}{T_{j}^{\mathsf{(ne)}}(t)}\times \frac{T_{j}^{\mathsf{(ne)}}(t)}{t}\times \left(1 - \frac{N_j^{(\mathsf{ccaf_{-i}})}(t)}{N_j^{(\mathsf{cca})}(t)}\right) }{\frac{I_{j}(t)}{t}}                           
\end{align*}
Using the earlier definitions, we have 
\begin{align}
\overline{\tau}_j^{(i)} &= \frac{\beta_j \times b_j \times q_j \times (1 - \alpha_j^{(-i)})}{1 - q_j + q_j \times b_j}\label{eqn:1}
\end{align}
where $\alpha_j^{(-i)}$ is the probability of CCA failure of node $j$ only due to the nodes in $\Omega_j \setminus \Omega_i$, and is derived below.

\subsubsection{Computation of $\alpha_j^{(-i)}$}

Using the notation introduced earlier, we can write 
\begin{align*}
   \alpha_j^{(-i)} &= \lim_{t\rightarrow \infty} \frac{N_j^{(\mathsf{ccaf_{-i}})}(t)}{N_j^{(\mathsf{cca})}(t)} \text{ a.s}\\
                   &= \lim_{t\rightarrow \infty} \frac{\frac{N_j^{(\mathsf{ccaf_{-i}})}(t)}{T_j^{(\mathsf{ne})}(t)}}{\frac{N_j^{(\mathsf{cca})}(t)}{T_j^{(\mathsf{ne})}(t)}} \\
\end{align*}
Applying the Renewal-Reward Theorem (RRT) \cite{kulkarni95modeling-stochastic-systems}, we have  
 \begin{align}                  
    \alpha_j^{(-i)} &= \frac{\frac{\mathbf{N}_j^{(\mathsf{ccaf_{-i}})}}{W_{j}}}{\frac{\mathbf{N}_j^{(\mathsf{cca})}}{W_{j}}} \\
    &= \frac{\mathbf{N}_j^{(\mathsf{ccaf_{-i}})}}{\mathbf{N}_j^{(\mathsf{cca})}}\label{eqn:alpha_j_i_exp}
\end{align} 
where $W_j$ is the mean time between the ends of transmission during a busy period of node $j$ (i.e., the mean renewal cycle length).  Now, by taking the reward to be the mean number of total CCAs by node $j$ in the renewal cycle, we can write 
\begin{align*}
   \mathbf{N}_j^{(\mathsf{cca})} &= \eta_j + (1 - \eta_j)c_j + (1 - \eta_j)(1 - c_j)(\beta_j T_j^{(\mathsf{eff})} + \mathbf{N}_j^{(\mathsf{cca})})   
\end{align*}

Note that node $j$ performs CCAs at the rate of $\beta_j$ for the entire dilated transmission period $T_j^{(\mathsf{eff})}$ i.e., a total of $\beta_j T_j^{(\mathsf{eff})}$ CCAs all of which fail, and the cycle continues after this period. Upon rearranging terms, we have
\begin{align}\label{eqn:N_cca_eqn}
     \mathbf{N}_j^{(\mathsf{cca})} &= \frac{\eta_j + (1 - \eta_j)c_j + (1 - \eta_j)(1 - c_j)\beta_j T_j^{(\mathsf{eff})} }{1 - (1 - \eta_j)(1 - c_j)}
\end{align}
Again, let the reward be the mean number of failed CCA attempts by node $j$ in a renewal cycle due to transmissions by nodes in $\Omega_j \setminus \Omega_i$, say $\mathbf{N}_j^{(\mathsf{ccaf_{-i}})}$, which we can write, using \textbf{(S7)} as
\begin{equation}\label{eqn:ccafail_node_j}\begin{split}
\mathbf{N}_j^{(\mathsf{ccaf_{-i}})}&=\left(\frac{\displaystyle{\sum_{k\in \Omega_j \cap \Omega_i} \overline{\tau}_k^{(j)}}}{\beta_j + \displaystyle{\sum_{l\in \Omega_j} \overline{\tau}_l^{(j)}}} \right)(1 - c_j)\left(\mathbf{N}_j^{(\mathsf{ccaf_{-i}})}\right) + \\ &\left(\frac{\displaystyle{\sum_{k\in \Omega_j \setminus \Omega_i} \overline{\tau}_k^{(j)}}}{\beta_j + \displaystyle{\sum_{l\in \Omega_j} \overline{\tau}_l^{(j)}}} \right)
(1 - c_j)\left(\beta_j T_{\mathsf{tx}} + \mathbf{N}_j^{(\mathsf{ccaf_{-i}})}\right)
\end{split}\end{equation}
Clearly, in case of a simultaneous channel sensing event involving node $j$, the cycle ends with zero reward. Therefore, the terms in Equation~\eqref{eqn:ccafail_node_j} only consider the case where node $j$ is not involved in simultaneous channel sensing. The first term in \eqref{eqn:ccafail_node_j} corresponds to transmission attempts by nodes in $\Omega_j \cap \Omega_i$; this results in zero reward, and the cycle continues thereafter. The second term accounts for transmission attempts by nodes in $\Omega_j \setminus \Omega_i$; this results in a total reward of $\beta_j T_{\mathsf{tx}}$ and the cycle continues thereafter. Upon rearranging terms, we have
\begin{align*}
 \mathbf{N}_j^{(\mathsf{ccaf_{-i}})}  &= \frac{\left(\frac{\displaystyle{\sum_{k\in \Omega_j \setminus \Omega_i} \overline{\tau}_k^{(j)}}}{\beta_j + \displaystyle{\sum_{l\in \Omega_j} \overline{\tau}_l^{(j)}}} \right)(1 - c_j)\left(\beta_j T_{\mathsf{tx}}\right)}{1 - (1 - \eta_j)(1 - c_j)}
\end{align*}
Hence from \ref{eqn:alpha_j_i_exp}, the CCA failure probability of node $j$ only due to the nodes in $\Omega_j \setminus \Omega_i$ is given by 
\begin{align}
  \alpha_j^{(-i)} &= \frac{\left(\frac{\displaystyle{\sum_{k\in \Omega_j \setminus \Omega_i} \overline{\tau}_k^{(j)}}}{\beta_j + \displaystyle{\sum_{l\in \Omega_j} \overline{\tau}_l^{(j)}}} \right)(1 - c_j)\left(\beta_j T_{\mathsf{tx}}\right)}{\eta_j + (1 - \eta_j)c_j + (1 - \eta_j)(1 - c_j)\beta_j T_j^{(\mathsf{eff})}}\label{eqn:6}
\end{align}

\remark For networks with no hidden nodes, $\alpha_j^{(-i)} = 0$.

\subsubsection{The dilated activity period, $T_i^{(\mathsf{eff})}$}
\label{subsubsec:Ti-eff}
As mentioned earlier, the length of the dilated activity period as perceived by a node depends upon the set of nodes in its CS range. We propose two different models for calculating $T_i^{(\mathsf{eff})}$.
\begin{itemize}
    \item[1. ] \textbf{$M/D/\infty$ Model :} Since each node has a different set of neighbours, we can make the following approximation to simplify our analysis:
\begin{description}
\item[(A1)] The mean length of Dilated Activity Period as perceived by node $i$ is equal to the length of mean busy period of an $M/D/\infty$ queue where the deterministic service time is equal to a single transmission period $T_{\mathsf{tx}}$, and the arrival process is approximated as a Poisson process having rate equal to the aggregate transmission initiation rate of \emph{all} the nodes in the CS range of node $i$. 
\end{description}
Figure~\ref{fig:MDinf_approx} represents the $M/D/\infty$ model pictorially. Note that the above assumption is equivalent to saying that all nodes in $\Omega_i$ are hidden from each other, and hence this approximation results in a larger mean for the dilated activity period. 

Note that the transmission initiation rate of any node $j\in \Omega_i$, as perceived by node $i$, is $\overline{\tau}_j^{(i)}$. For any node $i$, let the
aggregate transmission initiation rate for nodes in $\Omega_i$ be
\begin{figure}[t]
	\begin{center}
		\includegraphics[scale=0.5]{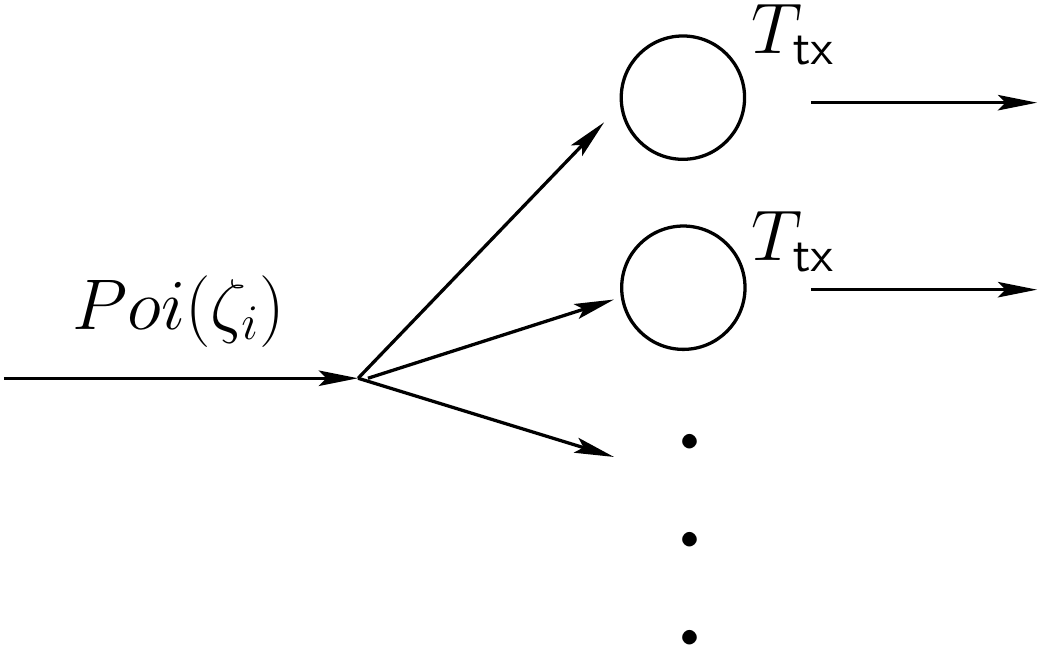}
		\caption{The $M/D/\infty$ approximation for the calculation of $T_i^{(\mathsf{eff})}$. The nodes in $\Omega_i$ are assumed to be hidden from each other and their transmission initiations are approximated by a Poisson process with a rate of $\zeta_i$ given by Equation~\eqref{eqn:zeta}. The service time for each of these transmissions is $T_{\mathsf{tx}}$.}
		\label{fig:MDinf_approx}
	\end{center}
\end{figure}
\begin{eqnarray}
\zeta_i = \sum_{j\in \Omega_i}\overline{\tau}_j^{(i)} 
\label{eqn:zeta}
\end{eqnarray}
Assuming node $j\in \Omega_i$ has started transmission and another node $k \in \Omega_i$ starts at time $u < T_{\mathsf{tx}}$, the expression for $T_i^{(\mathsf{eff})}$ can be written recursively as 
\begin{eqnarray}
\textstyle{T_i^{(\mathsf{eff})}} &\textstyle{=}& \textstyle{T_{\mathsf{tx}}\exp\{-\zeta_i T_{\mathsf{tx}}\} + \int_0^{T_{\mathsf{tx}}}(u+T_i^{(\mathsf{eff})})\zeta_i\exp\{-\zeta_i u \}du} \nonumber \\
&\textstyle{=}& \frac{1}{\zeta_i}(\exp\{\zeta_i T_{\mathsf{tx}}\} - 1)
\end{eqnarray}
 \item[2. ] \textbf{Boorstyn et al. \cite{boorstyn-etal87CSMA-throughput-analysis} Model:}
     For node $i$ in Figure~\ref{fig:t_eff}, Figure~\ref{fig:boorstyn_model} shows its evolution process in conditional time. We see that when nodes in $\Omega_i$ attempt packet transmission, then we have a dilated transmission period, $T_i^{(\mathsf{eff})}$ as perceived by Node $i$.
\begin{figure*}[t]
\begin{center}
	\includegraphics[scale=0.5]{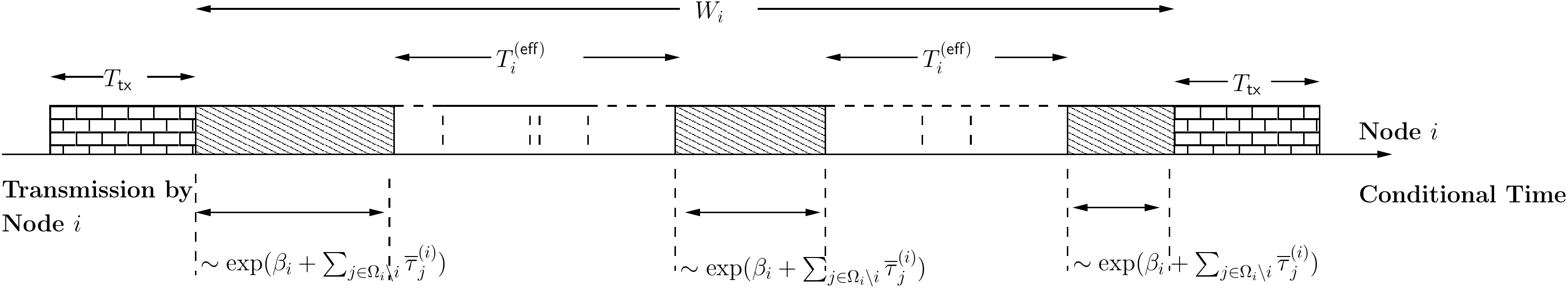}
	\caption{Evolution Process of Node $i$ in conditional time. We see that when nodes in $\Omega_i$ attempt packet transmission, we have a dilated transmission period, $T_i^{(\mathsf{eff})}$ as perceived by Node $i$.}
	\label{fig:boorstyn_model}
\end{center}
\end{figure*}

\begin{figure}[h]
\begin{center}
	\includegraphics[scale=0.3]{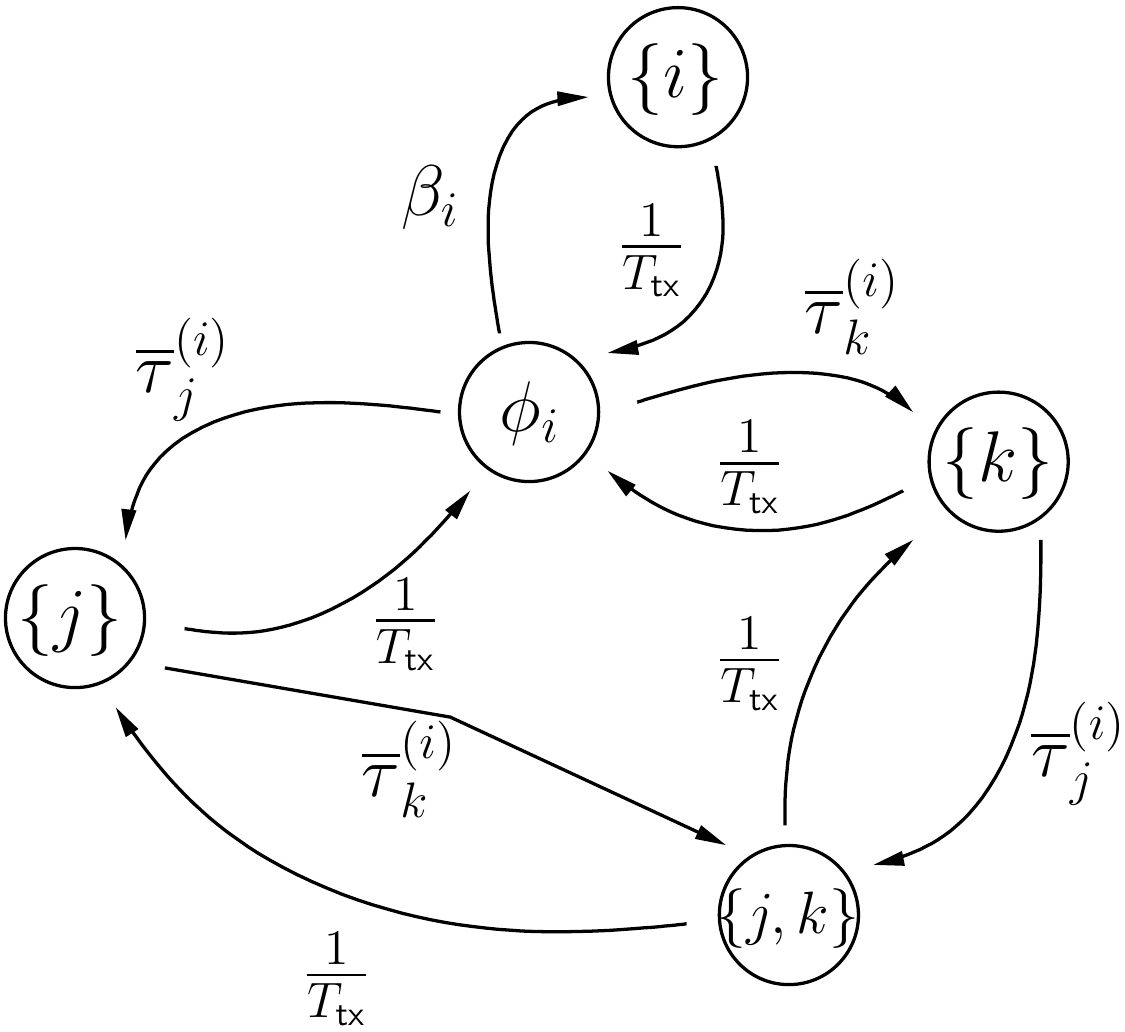}
	\caption{State transition diagram for computing $T_i^{(\mathsf{eff})}$ for the nodes in Figure~\ref{fig:t_eff}}
	\label{fig:state_diagram}
\end{center}
\end{figure}
Here we will use a CTMC model as suggested by Boorstyn et al. \cite{boorstyn-etal87CSMA-throughput-analysis} for determining $T_i^{(\mathsf{eff})}$ during times when node $i$ is non empty. Let $\phi_{i}$ be the state when all the nodes in $\Omega_i$ as well as node $i$ are in backoff. Then the state transition diagram for computing $T_i^{(\mathsf{eff})}$ for the nodes in Figure~\ref{fig:t_eff} is shown in Figure~\ref{fig:state_diagram}. Note that return to $\phi_{i}$ are renewal instants since the transition out of $\phi_{i}$ are all exponentially distributed. By the \emph{insensitivity property} explained in \cite{boorstyn-etal87CSMA-throughput-analysis}, the steady state probability of being in state $\phi_{i}$, say $\pi_{\phi_{i}}$, does not depend on the distribution of packet transmission time except by its mean. So we can take the packet service times to be exponentially distributed. Hence $\pi_{\phi_{i}}$ can be obtained as in \cite{boorstyn-etal87CSMA-throughput-analysis}, i.e.,

\begin{align}\label{eqn:phi_eqn_1}
  \pi_{\phi_{i}} &= \left[\sum_{D_{i}\in \mathscr{D}_i} \prod_{j \in D_{i}} \overline{\tau}_j^{(i)} T_{\mathsf{tx}}\right]^{-1}  
\end{align}
where $D_{i}$ is a set of nodes in $\Omega_i$ that are actively transmitting, and $\mathscr{D}_i$ is the collection of all such sets. Again, from the RRT we obtain $\pi_{\phi_i}$ as the ratio of the time the system is in state $\phi_{i}$ and the mean time between the ends of transmission during a busy period of Node~$i$, i.e., 
\begin{equation*}
    \resizebox{1.0\hsize}{!}{$\pi_{\phi_{i}} = \frac{\frac{1}{\beta_i + \displaystyle{\sum_{j \in \Omega_i} \overline{\tau}_j^{(i)}}}}{\frac{1}{\beta_i + \displaystyle{\sum_{j \in \Omega_i} \overline{\tau}_j^{(i)}}} + \frac{\beta_i}{\beta_i + \displaystyle{\sum_{j \in \Omega_i} \overline{\tau}_j^{(i)}}} \times T_{\mathsf{tx}} + \frac{\displaystyle{\sum_{j \in \Omega_i} \overline{\tau}_j^{(i)}}}{\beta_i + \displaystyle{\sum_{j \in \Omega_i} \overline{\tau}_j^{(i)}}} \times T_{i}^{(\mathsf{eff})}}$}
\end{equation*}
which yields
\begin{equation}\label{eqn:phi_eqn_2}   
    \pi_{\phi_{i}} = \frac{1}{1 + \beta_i \times T_{\mathsf{tx}} + \left(\displaystyle{\sum_{j \in \Omega_i} \overline{\tau}_j^{(i)}}\right)\times T_{i}^{(\mathsf{eff})}}
\end{equation}
Hence, from Equation~\eqref{eqn:phi_eqn_1} and \eqref{eqn:phi_eqn_2}, we can obtain $T_{i}^{(\mathsf{eff})}$ for each node $i$. Observe that the dilated activity period computed using the Boorstyn et al. model would be less conservative than the $M/D/\infty$ model. Also note that identification of $\mathscr{D}_{i}$ is an NP-complete combinatorial problem \cite{boorstyn-etal87CSMA-throughput-analysis}; however, efficient algorithms exist that handle networks of arbitrary topology and moderate size (50-100 nodes).  
\end{itemize}

\subsubsection{Probability of CCA failure, $\alpha_i$}
The probability of CCA failure is the probability of the occurrence of at least one transmitting node in the CS range of node $i$, given that node $i$ performs a CCA. Defining $N_i^{(\mathsf{cca})}(t)$ and $N_i^{(\mathsf{f})}(t)$ as the total number of CCAs and number of failed CCAs in the interval $(0,t]$, respectively, we have:
\begin{eqnarray*}
\alpha_i = \lim_{t \to \infty}\frac{N_i^{(\mathsf{f})}(t)}{N_i^{(\mathsf{cca})}(t)} \text{ a.s }= \lim_{t \to \infty}\frac{\frac{N_i^{(\mathsf{f})}(t)}{T_{i}^{\mathsf{(ne)}}(t)}}{\frac{N_i^{(\mathsf{cca})}(t)}{T_{i}^{\mathsf{(ne)}}(t)}}
\end{eqnarray*}
Applying Renewal-Reward Theorem (RRT) \cite{kulkarni95modeling-stochastic-systems}, we get
\begin{eqnarray}
\alpha_i &=&\frac{\frac{\mathbf{N}_i^{(\mathsf{f})}}{W_i}}{\frac{\mathbf{N}_i^{(\mathsf{cca})}}{W_i}}\nonumber\\
&=&\frac{\mathbf{N}_i^{(\mathsf{f})}}{\mathbf{N}_i^{(\mathsf{cca})}}
\label{eqn:alpha_rates}
\end{eqnarray}
where $\mathbf{N}_i^{(\mathsf{f})}$ is the mean number of failed CCAs and $\mathbf{N}_i^{(\mathsf{cca})}$ is the mean number of total CCAs in a cycle. $W_i$ is the mean renewal cycle length.

Computation of $\mathbf{N}_i^{(\mathsf{cca})}$ is as in Equation~\eqref{eqn:N_cca_eqn}. Now the mean number of failed CCA attempts by node $i$ in a renewal cycle, i.e., $\mathbf{N}_i^{(\mathsf{f})}$, can be written using \textbf{(S7)} as
\begin{eqnarray*}
\mathbf{N}_i^{(\mathsf{f})} &=& (1-\eta_i)(1-c_i)\Big(\beta_i T^{(\mathsf{eff})}_i+\mathbf{N}_i^{(\mathsf{f})}\Big)
\end{eqnarray*}
Rearranging and using Equation~\eqref{eqn:alpha_rates}, the CCA failure probability is given by
\begin{eqnarray}
\alpha_i &=& \frac{(1-\eta_i)(1-c_i)\beta_i T^{(\mathsf{eff})}_i}{\eta_i + (1-\eta_i)c_i + (1-\eta_i)(1-c_i)\beta_i T^{(\mathsf{eff})}_i}
\label{eqn:alpha}
\end{eqnarray}

\remark Note that for no hidden node network, $T^{(\mathsf{eff})}_i = T_{\mathsf{tx}}$.

\subsubsection{Packet failure probability, $\gamma_i$}
\label{subsubsec:packet-failure}
A transmitted packet can fail to be decoded by its intended receiver due to a collision, or due to noise. We do not take packet capture into account here (see \textbf{(S3)}). Define $M_i(t)$ as the total number of transmissions in interval $(0,t]$, $M_i^{(\mathsf{c})}(t)$  as total number of collisions in interval $(0,t]$ and $l_i$ as the probability of \emph{data} packet error (known by \textbf{(S5)}) on the link between node $i$ and its receiver $r(i)$ due to noise. Recall that ACKs are not corrupted by PHY noise (see \textbf{(S6)}). The probability of packet failure, $\gamma_i$, is defined as:
\begin{eqnarray*}
\gamma_i&=&\Bigg(\lim_{t \to \infty}\frac{M_i^{(\mathsf{c})}(t)}{M_i(t)}\Bigg)+\Bigg(1-\lim_{t \to \infty}\frac{M_i^{(\mathsf{c})}(t)}{M_i(t)}\Bigg)l_i
\text{ a.s }
\end{eqnarray*}
Dividing each term in limit by $T_{i}^{\mathsf{(ne)}}(t)$ and applying RRT, we get
\begin{eqnarray*}
\gamma_i &=& \Bigg(\frac{\frac{\mathbf{M}_i^{(\mathsf{c})}}{W_i}}{\frac{\mathbf{M}_i}{W_i}}\Bigg) + \Bigg(1 - \frac{\frac{\mathbf{M}_i^{(\mathsf{c})}}{W_i}}{\frac{\mathbf{M}_i}{W_i}}\Bigg)l_i
\end{eqnarray*}
where $\mathbf{M}_i^{(\mathsf{c})}$ is the mean number of collided packets, and $\mathbf{M}_i$ is the mean number of transmitted packets in a renewal cycle. If we ignore packet discards, there is a single transmission by node $i$ in each renewal cycle so that $\mathbf{M}_i$ becomes unity. We can rewrite 
\begin{eqnarray}
\gamma_i &=& p_i +(1-p_i )l_i
\label{eqn:gamma}
\end{eqnarray}
where $p_i $ is the probability of packet collision and equals $\mathbf{M}_i^{(\mathsf{c})}$. Since $l_i$ is given for each node $i$, we need to compute probability of packet collision, $p_i$, to compute packet failure probability, $\gamma_i$.

We denote the set of nodes that can cause interference in successful reception at $r(i)$ by $I_{r(i)}$, which is composed of two sets:
\begin{center}
$C^{(1)}_{r(i)}=\{j\in\mathscr{N}: j\in\Omega_i\text{ and }j\in I_{r(i)}\}$\\
$C^{(2)}_{r(i)}=\{j\in\mathscr{N}: j \notin\Omega_i\text{ and }j \in I_{r(i)}\}$
\end{center}
such that
\begin{eqnarray*}
C^{(1)}_{r(i)}\cap C^{(2)}_{r(i)}&=&\emptyset\\ 
C^{(1)}_{r(i)}\cup C^{(2)}_{r(i)}&=&I_{r(i)} 
\end{eqnarray*}
Note that the receiver $r(i)$ is in $C^{(1)}_{r(i)}$ (see Figure~\ref{fig:interference}).
\begin{figure}[ht]
\begin{center}
	\includegraphics[scale=0.45]{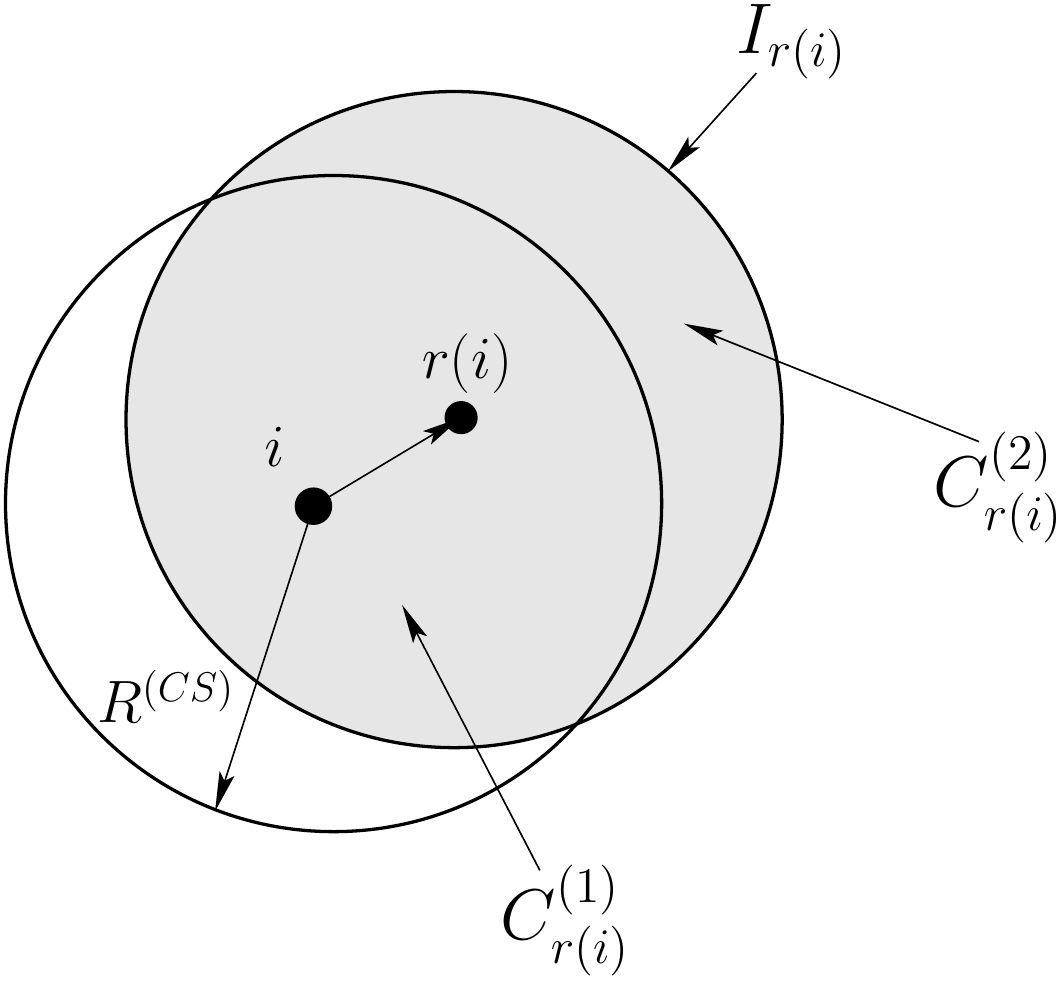}
    \caption{Interference region around the receiver. $R^{(CS)}$ is the CS range which is assumed to be equal for all nodes by \textbf{(S9)}.}
    \label{fig:interference}
	\vspace{-5mm}
\end{center}
\end{figure}

Let the reward be the mean number of collisions of node $i'$s packets in a renewal cycle, say, $\mathbf{M}_i^{(\mathsf{c})}$ , which can be written as
\begin{eqnarray*}
\mathbf{M}_i^{(\mathsf{c})} &=&  (1-\eta_i)(1-c_i)\mathbf{M}_i^{(\mathsf{c})} + \overline{\mathbf{M}}_i^{(\mathsf{c})}
\end{eqnarray*}
where $\overline{\mathbf{M}}_i^{(\mathsf{c})}$ is mean number of collisions of node $i'$s packets in a renewal cycle when node $i$ attempts a packet transmission or does a simultaneous channel sensing. On rearranging, the packet collision probability at node
$i$ is given by
\begin{eqnarray}
p_i &=& \frac{\overline{\mathbf{M}}_i^{(\mathsf{c})}}{1-(1-\eta_i)(1-c_i)}
\label{eqn:p}
\end{eqnarray}
To compute $\overline{\mathbf{M}}_i^{(\mathsf{c})}$, we first need to find the probability that nodes in a given set are not transmitting using a product approximation due to the unavailability of joint distribution of processes.  
Let $\overline{h}_i$ be the fraction of time node $i$ is not transmitting (unconditional) and is equal to 
\begin{eqnarray*}
\overline{h}_i &=& \lim_{t\rightarrow \infty}\frac{I_i(t)}{t} \quad \text{ a.s}\\
               &=& (1-q_i) + q_ib_i 
\end{eqnarray*}
The expression for $\overline{\mathbf{M}}_i^{(\mathsf{c})}$ can be found as the sum of several components as follows:
\begin{description}[leftmargin=0cm]
\item[1.] The first term accounts for the fact that Node $i$ started transmitting in the presence of at least one transmission by the hidden nodes set, i.e., $C^{(2)}_{r(i)}$. 
\begin{eqnarray*}
R_i^{(1)} &=& \eta_i\Bigg(1-\prod_{j\in C^{(2)}_{r(i)}}\overline{h}_j\Bigg)
\end{eqnarray*}
\item[2.] The second term accounts for the scenario when Node $i$ started its transmission as a simultaneous channel sensing with a node in $\Omega_i$ in the presence of at least one transmission by the hidden nodes set, i.e., $C^{(2)}_{r(i)}$.. 
\begin{eqnarray*}
R_i^{(2)} &=& (1-\eta_i)c_i\Bigg(1-\prod_{j\in C^{(2)}_{r(i)}}\overline{h}_j\Bigg)
\end{eqnarray*}
\item[3.] The third term considers the case where Node $i$ starts transmitting in the absence of any ongoing transmission by a hidden node but it encounters a simultaneous transmission by a node in $C^{(1)}_{r(i)}$ anywhere in the corresponding $12~T_s$ period and/or a transmission by a node in $C^{(2)}_{r(i)}$, i.e., a hidden node anywhere in its activity period $T_{\mathsf{tx}}$.
\begin{eqnarray*}
\scriptstyle{R_i^{(3)}} &\scriptstyle{=}& \scriptstyle{\eta_i\Bigg(\prod_{j\in C^{(2)}_{r(i)}}\overline{h}_j\Bigg)\Bigg(1-\exp \Bigg\{-12~T_s\Bigg( \displaystyle{\sum_{j\in C^{(1)}_{r(i)}}}\overline{\tau}_j^{(i)} \Bigg) \Bigg\}.} \\
&& \scriptstyle{\exp \Bigg\{-T_{\mathsf{tx}}\Bigg(\displaystyle{\sum_{j\in C^{(2)}_{r(i)}}}\overline{\tau}_j\Bigg)\Bigg\}\Bigg) }
\end{eqnarray*}
where $\overline{\tau}_j$ is the rate of successful CCA attempts of node $j$ over non-transmitting period, $I_j(t)$ and is equal to 
\begin{align*}
    \overline{\tau}_j &= \lim_{t \rightarrow \infty }\frac{N^{(\mathsf{cca})}_j(t) - N^{(\mathsf{f})}_j(t)}{I_j(t)} \quad \text{ a.s}\\
                      &= \frac{\beta_j b_j q_j (1 - \alpha_j)}{1 - q_j + q_jb_j}
\end{align*}
\item[4.] The fourth term says that Node $i$ started its transmission as a simultaneous channel sensing with a node in $C^{(1)}_{r(i)}$ in the absence of any hidden node. This event surely ends up in a collision at the receiver $r(i)$.
\begin{eqnarray*}
R_i^{(4)} &=& \Bigg(\frac{\displaystyle{\sum_{j\in C^{(1)}_{r(i)}}\overline{\tau}_j^{(i)}}}{\beta_i+\displaystyle{\sum_{j\in \Omega_i}\overline{\tau}_j^{(i)}}}\Bigg)c_i\Bigg(\prod_{j\in C^{(2)}_{r(i)}}\overline{h}_j\Bigg) 
\end{eqnarray*}
\item[5.] The final term says that Node $i$ started its transmission as a simultaneous channel sensing with a node in $\Omega_i \setminus C^{(1)}_{r(i)}$ in the absence of any hidden node but it encounters a simultaneous transmission by a node in $C^{(1)}_{r(i)}$ and/or by a node in $C^{(2)}_{r(i)}$, i.e., a hidden node.
\begin{align*}
R_i^{(5)} &= \Bigg(\frac{\displaystyle{\sum_{j\in \Omega_i\setminus C^{(1)}_{r(i)}}\overline{\tau}_j^{(i)}}}{\beta_i+\displaystyle{\sum_{j\in \Omega_i}\overline{\tau}_j^{(i)}}}\Bigg)c_i\Bigg(\prod_{j\in C^{(2)}_{r(i)}}\overline{h}_j\Bigg) \nonumber\\
& \bigg(1-\exp \Bigg\{-12~T_\mathsf{s}\Bigg( \displaystyle{\sum_{j\in C^{(1)}_{r(i)}}}\overline{\tau}_j^{(i)} \Bigg) \Bigg\} \nonumber\\
&\exp \Bigg\{-T_{\mathsf{tx}}\Bigg(\displaystyle{\sum_{j\in C^{(2)}_{r(i)}}}\overline{\tau}_j\Bigg)\Bigg\}\Bigg) 
\end{align*}
\end{description}
The overall packet collision probability is given by 
\begin{eqnarray}
p_i &=& \frac{R_i^{(1)}+R_i^{(2)}+R_i^{(3)}+R_i^{(4)}+R_i^{(5)}}{\eta_i+(1-\eta_i)c_i}									
\label{eqn:p_hidden}
\end{eqnarray}
The packet failure probability, $\gamma_i$, can now be calculated using Equations~\eqref{eqn:p_hidden} and \eqref{eqn:gamma}. 

\remark Note that for no hidden node network $C^{(2)}_{r(i)} = \phi$ and $\prod_{j\in C^{(2)}_{r(i)}}\overline{h}_j = 1$
 
\subsubsection{Average Service Rate ($\sigma_i$)}
Each packet that reaches the HOL position in the queue at a node can have multiple backoffs and transmissions before it is successfully transmitted or discarded. We define $\overline{Z}_i$ as the average time spent in backoff, and $\overline{Y}_i$ as the average transmission time until the packet is successfully transmitted or discarded at the MAC layer. Then the average service rate, $\sigma_i$, at node $i$ can be calculated as:
\begin{eqnarray}
\frac{1}{\sigma_i}=\overline{Z}_i + \overline{Y}_i
\label{eqn:sigma}
\end{eqnarray}
Using the default values from the standard, the mean backoff durations can be calculated (shown in Table~\ref{tab:mean_backoff_durations}).
\begin{table}[b]
\begin{center}
{\scriptsize
\begin{tabular}{ccc}
\hline
\hline
\textbf{Backoff} & \textbf{Mean Backoff Duration} & \textbf{Mean Backoff Duration}\\
\textbf{Exponent} & \textbf{with successful CCA($T_\mathsf{s}$)} & \textbf{with failed CCA($T_\mathsf{s}$)}\\
\hline
\hline
$3$ & $70+20$ & $70+8$\\
\hline
$4$ & $150+20$ & $150+8$\\
\hline
$5$ & $310+20$ & $310+8$\\
\hline
\hline
\end{tabular}
}
\caption{Mean backoff durations in symbol times, $T_\mathsf{s}$}
\label{tab:mean_backoff_durations}
\end{center}
\end{table}
Define the following quantities:
\begin{align*}
\scriptstyle{\overline{B}_i} &\scriptstyle{=} \scriptstyle{\left(70+8+158\alpha_i + 318\alpha_i^2 + 318\alpha_i^3 + 318\alpha_i^4\right)}\\
\scriptstyle{T_i^{(1)}}&\scriptstyle{=}\scriptstyle{\Bigg(\frac{(70+8)(1-\alpha_i)}{(1-\alpha_i^5)}+\frac{236\alpha_i(1-\alpha_i)}{(1-\alpha_i^5)}+\frac{554\alpha_i^2(1-\alpha_i)}{(1-\alpha_i^5)}}\\
&\scriptstyle{+\frac{872\alpha_i^3(1-\alpha_i)}{(1-\alpha_i^5)}+\frac{1190\alpha_i^4(1-\alpha_i)}{(1-\alpha_i^5)}\Bigg) }\\
\scriptstyle{T_i^{(2)}}&\scriptstyle{=}\scriptstyle{(78+158+318+318+318)}
\end{align*}
where $\overline{B}_i$ refers to the mean backoff duration until the packet is transmitted or discarded due to successive CCA failures, $T_i^{(1)}$ has the interpretation of the mean backoff duration given that the packet transmission was successful, and $T_i^{(2)}$ is the mean time spent in backoff given that the packet was discarded due to successive CCA failures. Then the quantities $\overline{Z}_i$ and $\overline{Y}_i$ can be calculated as:
\begin{eqnarray*}
\scriptstyle{\overline{Z}_i}&\scriptstyle{=}&\scriptstyle{\alpha_i^5T_i^{(2)}+(1-\alpha_i^5)[T_i^{(1)}+\gamma_i(\alpha_i^5T_i^{(2)}+(1-\alpha^5)} \\
&&\scriptstyle{[T_i^{(1)}+\gamma_i(\alpha_i^5T_i^{(2)}+(1-\alpha_i^5)[T_i^{(1)}+\gamma_i(\alpha_i^5T_i^{(2)}+(1-\alpha_i^5)T_i^{(1)})])])]}\\
\scriptstyle{\overline{Y}_i}&\scriptstyle{=}&\scriptstyle{(1-\alpha_i^5)[T_{\mathsf{tx}}+\gamma_i(1-\alpha_i^5)[T_{\mathsf{tx}}+\gamma_i(1-\alpha_i^5)[T_{\mathsf{tx}}+\gamma_i(1-\alpha_i^5)T_{\mathsf{tx}}]]]}
\end{eqnarray*}

\subsubsection{Aggregate Arrival Rate ($\nu_i$), Goodput ($\theta_i$) and Discard Probability ($\delta_i$) for a Node}
The arrival process at each node consists of packets which are generated at the same node (if it is a sensor) and the packets to be forwarded. The aggregate arrival rate at a node $i$ can be written as:
\begin{eqnarray}
\nu_i=\lambda_i+\sum_{k\in\mathcal{P}_i}\theta_k
\label{eqn:nu}
\end{eqnarray}
where $\lambda_i$ is the packet generation rate at sensor node $i$, $\mathcal{P}_i$ is the set of all its in-neighbours, and $\theta_i$ is the rate of packet transmission by the node which are \emph{successfully received} at $r(i)$ (known as goodput). An enqueued packet at a node is successfully received by the intended receiver in a manner shown in Figure~\ref{fig:goodput}.
\begin{figure}[h]
\begin{center}
	\includegraphics[scale=0.5]{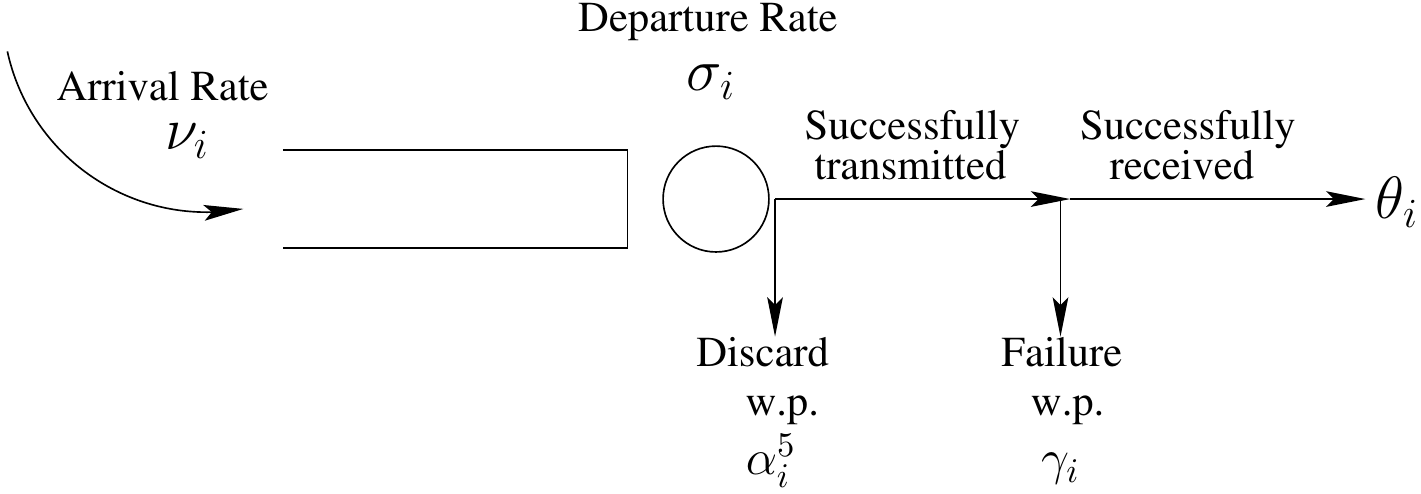}
	\caption{The goodput of a node is defined as the rate of successfully received packets by its receiver.}
	\label{fig:goodput}
	\vspace{-5mm}
\end{center}
\end{figure}

Assuming that the queueing system is stable and has a steady-state solution, the goodput $\theta_i$ is given by
\begin{eqnarray}
\theta_i=\nu_i(1-\delta_i)
\label{eqn:theta_arrival}
\end{eqnarray}
where $\delta_i$ is the probability of discarding a packet due to consecutive CCA failures, or successive failed retransmission attempts, and is given by
\begin{eqnarray}
\scriptstyle{\delta_i=\alpha_i^5+(1-\alpha_i^5)\gamma_i\Big[\alpha_i^5+(1-\alpha_i^5)\gamma_i\Big[\alpha_i^5+(1-\alpha_i^5)\gamma_i\Big[\alpha_i^5+(1-\alpha_i^5)\gamma_i\Big]\Big]\Big]}
\label{eqn:delta}
\end{eqnarray}
Note that if the queue at node $i$ is saturated, then the goodput $\theta_i$ is equal to $\sigma_i$. Note that while calculating $\sigma_i$, we have taken the packet discards into account. 

\subsubsection{The node non-empty probability, $q_i$}\label{non_empty_prob}
To find the expression for $q_i$, assuming that all the arriving packets reach the HOL position (i.e., no tail drops) and applying Little's Theorem, we get
\begin{eqnarray}
q_i=\frac{\nu_i}{\sigma_i}
\label{eqn:q_arrival}
\end{eqnarray}
Further, for a saturated node, the quantity $q_i$ is equal to $1$.

\subsubsection{Obtaining $b_i$ and $\beta_i$}
In order to find $b_i$, the fraction of time a node is in backoff provided it is non-empty, we embed a renewal process in conditional time where the renewal epochs are those instants at which the node enters the random backoff period after a packet transmission or packet discard. We use the RRT to find the expression for $b_i$ as 
\begin{eqnarray}
b_i=\frac{\overline{B}_i}{\overline{B}_i + (1-\alpha_i^5)T_{\mathsf{tx}}}
\label{eqn:calc_of_b}
\end{eqnarray}

The quantity $\beta_i$, the rate of CCA attempts in backoff times, is based on the backoff completion times irrespective of the transmission attempt of the packet since a packet retransmission is considered the same as a new packet transmission. Thus, for the calculation of $\beta_i$, it suffices to observe the process at node $i$ only in the backoff times. Figure~\ref{fig:backoff_times} shows the residual backoff process where after completion of the backoff duration, the node performs a CCA. 
\begin{figure}[t]
\begin{center}
\includegraphics[height=30mm, width=93mm]{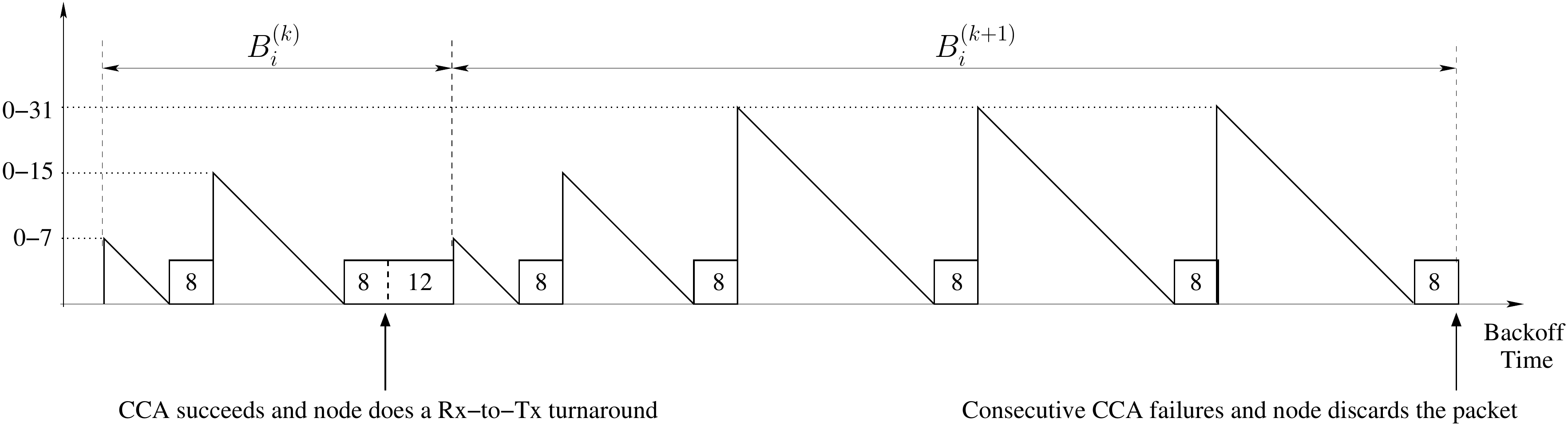}
\caption{Evolution of backoff periods in conditional (backoff) time. The $y$-axis is in units of a backoff slot ($20 ~T_\mathsf{s}$). The range indicates that the random backoff duration is uniformly distributed within it. The CCA duration, and Rx-to-Tx turnaround duration are 8 and 12 symbol times respectively, as indicated in the boxes.}
\label{fig:backoff_times}
\vspace{-5mm}
\end{center}
\end{figure}
Using the result of Kumar et al. \cite{kumar-etal04new-insights} with their collision probability replaced by our CCA failure probability, $\alpha_i$, the expression for $\beta_i$ can be written as 
\begin{eqnarray}
\beta_i=\frac{1+\alpha_i+\alpha_i^2+\alpha_i^3+\alpha_i^4}{\overline{B}_i}
\label{eqn:beta}
\end{eqnarray}

\section{Iterative Solution and Calculation of Performance Measures}
\label{sec:iterative-scheme}
\begin{figure}[ht]
\begin{center}
	\includegraphics[scale=0.38]{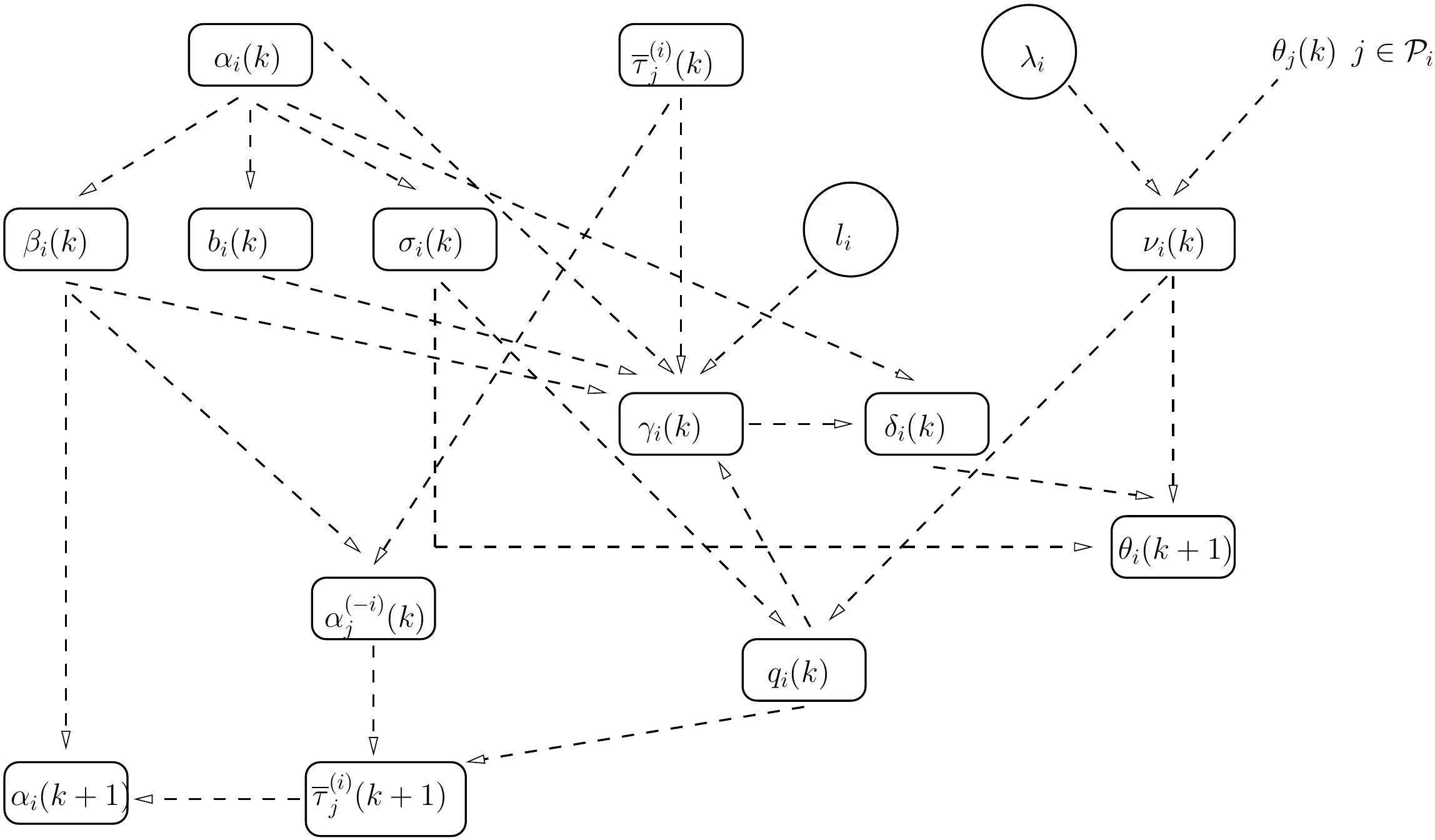}
	\caption{The global iteration scheme. The square boxes indicate the derived quantities and the round boxes indicate the known quantities.}
	\label{fig:iterations}
	\vspace{-8mm}
\end{center}
\end{figure}

\subsection{Global Iteration Scheme}\label{sec:iterations}
The global iteration scheme to solve the fixed point equations derived above is shown in Figure~\ref{fig:iterations}. We start with a vector $\{0,10,\lambda_i,l_i\}$ corresponding to $\{\alpha_i,\overline{\tau}_j^{(i)},\nu_i,l_i\}$ for each node $i$ and node $j \in \Omega_i$, and repeat the procedure until the quantities converge. 

\subsection{Existence of a fixed point}
A careful look at the derivation of the fixed point equations in Section~\ref{subsec:fixed-pt-derivation} reveals that the fixed point variables are $\{(\alpha_i,q_i)\}_{i=1}^N$, where $\alpha_i$ is the CCA failure probability at node $i$, and $q_i$ is the queue non-empty probability at node $i$. These variables are related via the set of equations~\eqref{eqn:1}, \eqref{eqn:6}, \eqref{eqn:phi_eqn_1}, \eqref{eqn:phi_eqn_2}, \eqref{eqn:alpha}, \eqref{eqn:gamma}, \eqref{eqn:p_hidden}, and \eqref{eqn:sigma}-\eqref{eqn:beta}. We observe that all the functions involved in the fixed point equations are continuous (since compositions, sums, and products of continuous functions are continuous, and also minimum of continuous functions is continuous). Hence, the fixed point equations can be represented as a \emph{continous} map from $[0,1]^{2N}$ to $[0,1]^{2N}$. Hence, the existence of a solution to the fixed point equations follows from Brouwer's fixed point theorem.

Proving the uniqueness of the fixed point, and convergence of the iterative procedure to the fixed point is out of the scope of this work. For a detailed discussion of the complexity associated with these proofs for general multi-hop topologies, see Section IV.D in \cite{jindal09}. To the best of our knowledge, the only work to have attempted a formal proof of the uniqueness of the fixed point in the context of multi-hop networks is \cite{marbach11}, but that too for a rather simplified version of CSMA policies. However, \emph{in our numerical experiments, the fixed point iterations always converged to a solution}. 


\subsection{End-to-End Delay Calculation}

We use the two parameter Whitt's QNA \cite{whitt83queueing-network-analyzer} to calculate the mean sojourn time at each node based on an approximate Moment Generating Function (MGF) of service time. We can write an expression for service time in a recursive manner by allowing infinite CCAs and retransmissions. To account for the discarded packets, we allow only a fraction of transmitted packets to join the next-hop neighbour's queue based on the discard probability at that node. We denote the service time at node $i$ by $S_i$, and let $B_i$ denote the length of random backoff duration, which is assumed to have an exponential distribution with rate $\beta_i$ for node $i$. Then, $S_i$ can be written as
\begin{eqnarray*}
S_i &=& \begin{cases}
			B_i + \tilde{S}_i & \text{w.p. } \alpha_i \\
			B_i + T_{\mathsf{tx}} & \text{w.p. } (1-\alpha_i)(1-\gamma_i) \\
			B_i + T_{\mathsf{tx}} + \tilde{S}_i & \text{w.p. } (1-\alpha_i)\gamma_i
		\end{cases}
\end{eqnarray*}
where $\tilde{S}_i$ is a random variable with the same distribution as $S_i$. The MGF of $B_i$, denoted by $M_{B_i}(z)$, is equal to $\frac{\beta_i}{z+\beta_i}$. Therefore, we can express the MGF $M_{S_i}(z)$ of service time $S_i$ as 
\begin{eqnarray*}
M_{S_i}(z) &=& \frac{\beta_i(1-\alpha_i)(1-\gamma_i)e^{-zT_{\mathsf{tx}}}}{z+\beta_i(1-\alpha_i)(1-\gamma_ie^{-zT_{\mathsf{tx}}})}.
\end{eqnarray*}
The first two moments of the service time, $\mathbb{E}(S_i)$ and $\mathbb{E}(S_i^2)$ can be calculated by differentiating the MGF $M_{S_i}(z)$:
\begin{eqnarray*}
\mathbb{E}(S_i) &=& -\frac{d(M_{S_i}(z))}{dz}\Bigg \vert_{z=0} \\
\mathbb{E}(S_i^2) &=& \frac{d^2(M_{S_i}(z))}{dz^2}\Bigg \vert_{z=0} 
\end{eqnarray*}
The QNA procedure commences from the leaf nodes and converges to the base-station in a sequential manner. We first calculate the squared coefficient of variance of service time, denoted by $c_{S_i}^2$, at each node, i.e.,
\begin{eqnarray*}
c_{S_i}^2 &=& \frac{\mathbb{E}(S_i^2)}{(\mathbb{E}(S_i))^2} - 1
\end{eqnarray*}
Ignoring the packet discards at the MAC layer, the net arrival rate $\Lambda_i$ at a node $i$ is given by:
\begin{eqnarray*}
\Lambda_i &=& \begin{cases}
					\lambda_i + \displaystyle{\sum_{k\in \mathcal{P}_i}}\theta_k & \text{for source-cum-relay nodes} \\
					\displaystyle{\sum_{k\in \mathcal{P}_i}}\theta_k & \text{for relay nodes}
			  \end{cases}
\end{eqnarray*}
For any node $i$, let $\rho_i$ be defined as
\begin{eqnarray*}
\rho_i &=& \begin{cases}
					\lambda_i\mathbb{E}(S_i) & \text{for leaf nodes} \\
					\Lambda_i\mathbb{E}(S_i) & \text{for internal nodes}
			  \end{cases}
\end{eqnarray*}
We now calculate the squared coefficient of variance for the departure process ($c_{D_i}^2$) for a node which uses the squared coefficient of variance for the interarrival times at that node ($c_{A_i}^2$). Here we include the probability of a packet discard which is equal to $\delta_i$.
\begin{eqnarray*}
c_{D_i}^2 &=& (1 - \delta_i)(1 + \rho_i^2(c_{S_i}^2 - 1) + (1-\rho_i^2)(c_{A_i}^2 - 1))
\end{eqnarray*}
Further, $c_{A_i}^2$ is calculated as 
\begin{eqnarray*}
c_{A_i}^2 &=& \frac{1}{\Lambda_i}\Bigg(\lambda_i + \displaystyle{\sum_{j\in \mathcal{P}_i}}\Lambda_jc_{D_j}^2\Bigg)
\end{eqnarray*}
Finally, the mean sojourn time at a node $i$ is given by
\begin{eqnarray}
\overline{\Delta}_i &=& \frac{\rho_i\mathbb{E}(S_i)(c_{A_i}^2 + c_{S_i}^2)}{2(1-\rho_i)} + \mathbb{E}(S_i)
\end{eqnarray}
The end-to-end mean packet delay for a source node $j$, provided that the set of nodes along the path from this node to the BS is $L_j$, is given by
\begin{eqnarray}
\Delta_j &=& \displaystyle{\sum_{i\in L_j}\overline{\Delta}_i}
\end{eqnarray}


\subsection{Packet Delivery Probability for each Source Node}
$p_i^{(del)}$ for each source node $i$ is defined as the fraction of generated packets at source node $i$ that reach the base station without any time bound. Let the set of nodes constituting the path from a source node $i$ to the BS be $L_i$. Assuming that the drop events are independent from node to node, the expression for $p_i^{(del)}$ is
\begin{eqnarray}
p_i^{(del)} &=& \prod_{j\in L_i}(1-\delta_j)
\end{eqnarray}
where $\delta_j$ is the packet discard probability of node $j$. 

\section{Discussion on Validity of the Fixed Point Approach}
\label{sec:dtmc-stability}
We recall that we developed the analysis under the premise that the system of queues is stable, and, hence, all the steady state quantities involved in the fixed point equations exist. The next question we ask is whether, having performed the above analysis, we can use the results to conclude that the system of queues is indeed stable, which would provide a consistency in the overall approach.

We proceed by modeling the CSMA/CA multihop network analyzed in Section~\ref{sec:analytical-model} as a Discrete Time Markov Chain (DTMC), and therefrom, deriving a sufficient condition for the stability of the network.  

\subsection{A DTMC Model}
\label{subsec:dtmc}
The system evolves \emph{synchronously} over slotted time (this is an idealization only for the analysis in this section), with the duration of each slot being 1 symbol time \footnote{Note that in a CSMA/CA network, all durations (e.g., backoff, transmission, CCA) are multiples of the symbol time}. We assume that the \emph{start instants of all backoffs are aligned with the slot boundaries}. Further, the external arrival process to each source node is assumed to be Poisson; hence the number of arrivals in successive slots are independent (independent increment property). The evolution of the queueing system can then be modeled as a DTMC embedded at these slot boundaries, with the state of each queue comprising of the following components:

\begin{enumerate}
\item Queue length at each node
\item Residual backoff at each node 
\item Residual CCA time at each node 
\item Backoff stage (i.e., the number of CCA attempts for the current packet) at each node
\item Transmission state (whether the node is transmitting a packet or not) at each node
\item Retransmission stage (the number of retransmissions so far for the current packet) at each node
\item Residual packet length when a node is transmitting (in symbol times)
\end{enumerate}

Note that except the queue lengths, all the other components of the state space are \emph{finite}, while the state space of the queue lengths is countable. Hence, the state space is countable. We denote the state at time step $n$ by $(\X(n),\Y(n))$, where $\X(n)=(X_1(n),\ldots,X_N(n))$ denotes the queue length process, and $\Y(n)$ denotes the rest of the components of the state taken together. Thus, $\X(n)\in\mathbb{Z}^n_+$, and $\Y(n)\in K_y$, where $K_y$ is a finite set of finite valued vectors.  

We adopt the following convention from \cite{panda} (see Figure 4 in Section~3.1 in \cite{panda}) for counting the queue lengths at each time step: all arrivals that occur during a slot are counted immediately before the end boundary of the slot, all packets that leave the queues during a slot are counted at the end boundary of the slot, and the queue lengths are computed immediately after the end boundary of the slot, so as to account for all the arrivals and departures during the slot. 

With this setup, the queueing system evolves as a DTMC over the state space $\mathbb{Z}^n_+\times K_y$. The transition probabilities of the DTMC are governed by the distribution of the arrival process, and the backoff distribution. Note that a similar DTMC model was proposed in \cite{kumar-etal04new-insights} for IEEE~802.11 CSMA/CA under saturation assumption.

Now note that the operation of the system starts with all the queues empty, i.e., in the state $(\X(0)=\0,\Y(0)=\0)$. \emph{We are, therefore, only interested in the communicating class of the DTMC containing the all-zero state}, $(\X(0)=\0,\Y(0)=\0)$. Let us denote this class by $\mathcal{C}_{0}$. 

\begin{proposition}
\label{prop:irreducible}
The class $\mathcal{C}_{0}$ is closed, and aperiodic. 
\end{proposition} 

\begin{proof}
See the Appendix.
\end{proof}
It follows from Proposition~\ref{prop:irreducible} that the DTMC evolving from the all-zero state is confined to the class $\mathcal{C}_{0}$. \emph{With abuse of notation, from now on, we denote by $(\X(n),\Y(n))$, the DTMC evolving from the all-zero state.}
\subsection{A Sufficient Condition for Network Stability}
\label{subsec:stability-conditions}
Since we have a DTMC, by ``stability'', we mean that the DTMC $(\X(n),\Y(n))$ is \emph{positive recurrent}. Then, since $(\X(n),\Y(n))$ is also aperiodic, it can be verified that the steady state rates introduced in Section~\ref{sec:analytical-model} \emph{exist}. Therefore, we shall derive conditions for positive recurrence of the DTMC.

Note that the queue non-empty probability of queue $i$, $1\leq i\leq N$, in the steady state is\footnote{under the condition of aperiodicity, this limit always exists} 
\begin{equation*}
q_i = \lim_{n\to \infty}\prob[X_i(n) > 0|\X(0)=\mathbf{0},\Y(0)=\mathbf{0}]
\end{equation*}

We have the following sufficient condition on $q_i$, $1\leq i\leq N$, for positive recurrence of the DTMC $(\X(n),\Y(n))$.

\begin{theorem}
\label{thm:stability}
If $\sum_{i=1}^N q_i\:<\:1$, then the DTMC $(\X(n),\Y(n))$ (evolving from the all-zero state) is positive recurrent. 
\end{theorem}
\begin{proof}
See the Appendix. 
\end{proof}

\gap
\noindent
\textbf{Discussion:}
We began this section by asking the question whether the results from our analytical model can be used to verify the initial assumption of stability. Since the analysis involves many approximations, evidently no definite answer can be given to this question. However, we have found from extensive simulations (see Section~\ref{sec:numerical-and-simulation-results} for details of the simulation procedure) that the solution to the fixed point analysis models $\sum_{i=1}^N q_i$ accurately (\emph{within an error of up to 10\%}) for external arrival rates of up to about 6 packets/sec (see Figure~\ref{fig:measure_sum_qi} for example). 
\begin{figure}[ht]
\begin{center}
	\includegraphics[scale=0.32]{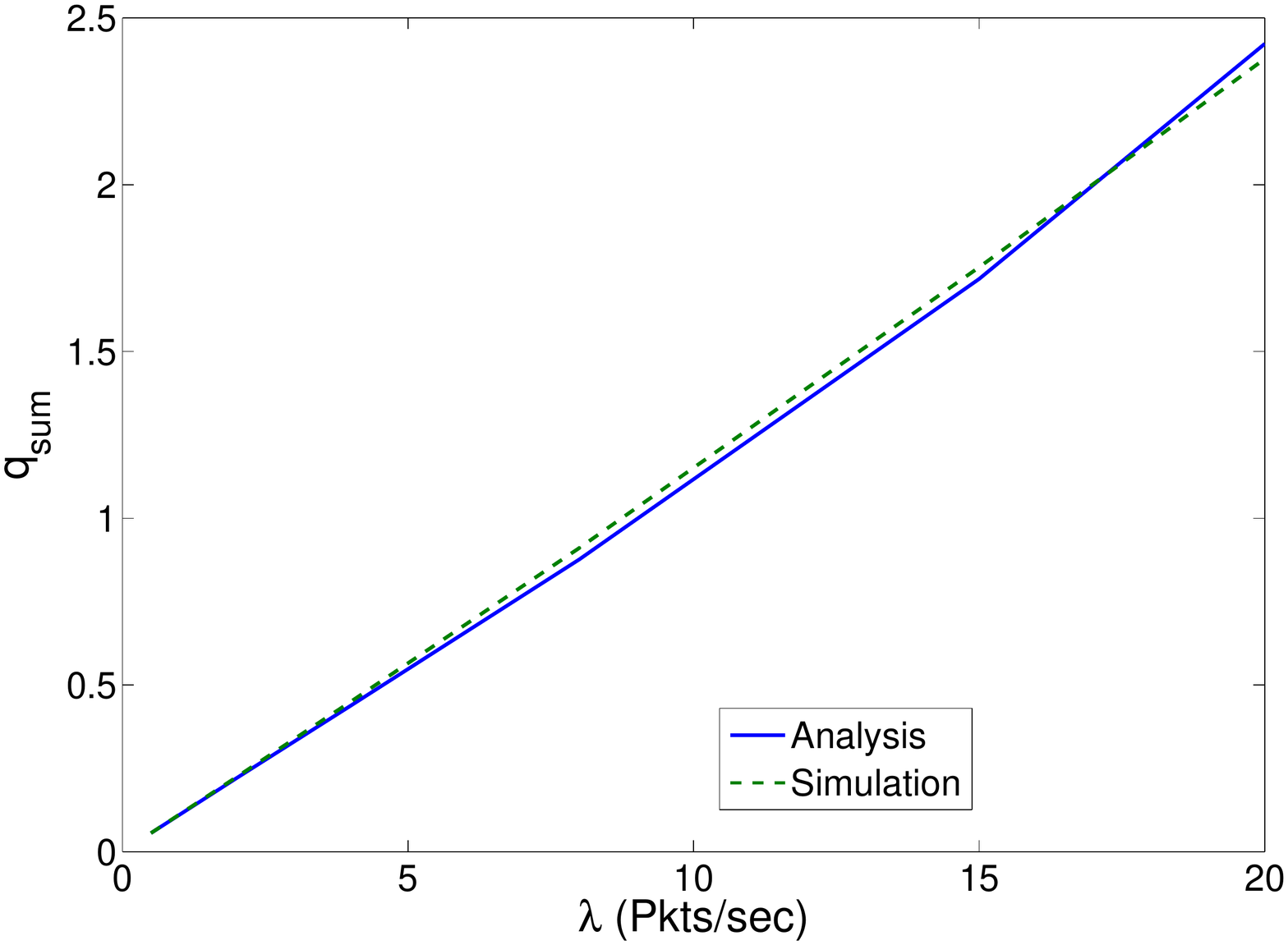}
	\caption{Plot of measure $\sum_{i=1}^N q_i$ for the topology in Figure~\ref{fig:hidden_10nodes}}
	\label{fig:measure_sum_qi}
	\vspace{-8mm}
\end{center}
\end{figure}
This observation along with Theorem~\ref{thm:stability} can support the claim that if $\sum_{i=1}^N q_i < 0.9$ in the solution to the fixed point analysis, then we can safely assume the network to be stable.


\section{Numerical and Simulation Results} \label{sec:numerical-and-simulation-results}

For the verification of our analytical model, we use QualNet (v4.5) simulator \cite{qualnet} with the default parameter values and a fixed payload size of $70$ bytes. However, the QualNet implementation is devoid of ACKs. Therefore, we compare results only for the ACK\emph{-less} scenarios (although the analysis permits the modeling of ACKs). We use the following simulation models in Qualnet to declare that a packet is in error:
\begin{enumerate}
      \item Collision and Link Error Model: If a receiver is receiving a packet from a node, and there is another transmission in the carrier sense range of the receiver, then the receiver concludes collision of all the packets that are intended for it. In addition, if a receiver $r(i)$ of node $i$ receives a packet successfully from node $i$ since there is no other packet transmission in the receiver's carrier sense range, then the received packet is concluded to be in error with probability $l_i$, the probability of data packet error on the link between node $i$ and the receiver $r(i)$ due to noise.
     \item Capture Model: If a receiver is receiving a packet from a node, and there is another transmission in the carrier sense range of the receiver, then the receiver computes the PER (from the SINR) of the packet that is meant for it, and rejects the packet if the computed PER is greater than a random number generated between 0 and 1.
\end{enumerate}
To increase the accuracy of simulation for each arrival rate, we generate the packets for 1500 seconds, and average the results over 25 repetitions with different random number seeds. 1500 seconds was chosen so that even at lower arrival rates ($\leq$ 2pkt/sec), sufficient number of packets would be generated from each source, thereby allowing the simulation to reach steady state.

\subsection{Simulation of an Example Network Topology}
The network shown in Figure~\ref{fig:hidden_10nodes} has hidden nodes (the dependency graph is also shown). Note that all the nodes are sources (some of which also serve as relays) with identical packet generation rates that are simultaneously increased. 
Figures~\ref{fig:measure_alpha_for_tree-n10-CS3-PER0.01}, \ref{fig:measure_gamma_for_tree-n10-CS3-PER0.01}, \ref{fig:measure_delta_for_tree-n10-CS3-PER0.01}, and \ref{fig:measure_pdel_for_tree-n10-CS3-PER0.01} are plots of measures $\alpha_i$, $\gamma_i$,  $\Delta_i$ and $p^{(\mathsf{del})}_i$ for the nodes in Figure~\ref{fig:hidden_10nodes}.  

\begin{figure}[ht]
\begin{center}
	\includegraphics[scale=0.32]{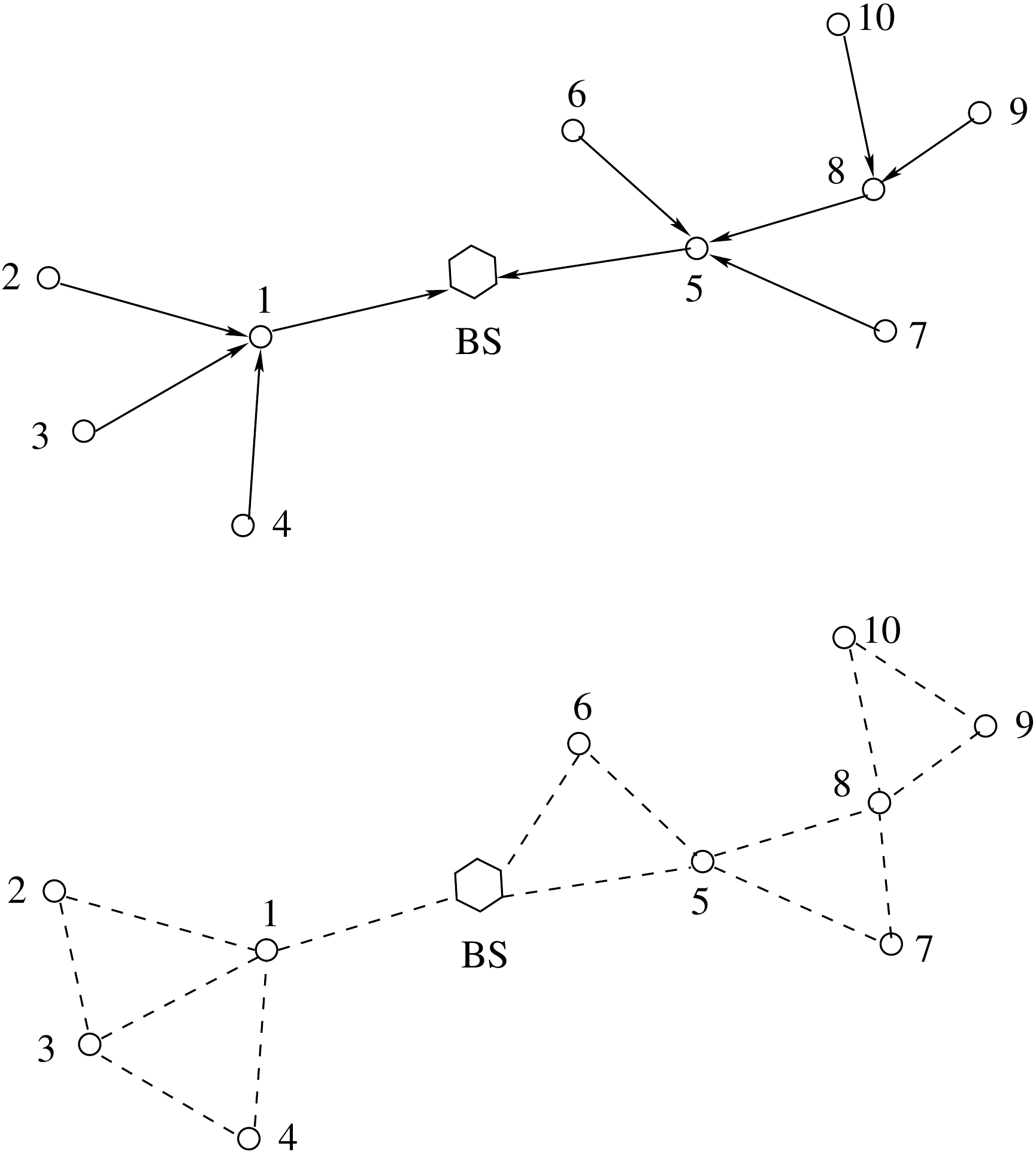}
	\caption{A 10 nodes TREE topology. Also shown is the dependency graph where the dotted lines connecting two nodes indicate that the nodes are in CS range of each other.}
	\label{fig:hidden_10nodes}
	\vspace{-8mm}
\end{center}
\end{figure}

\begin{figure}[ht]
\begin{center}
	\includegraphics[height=7cm, width=9.5cm]{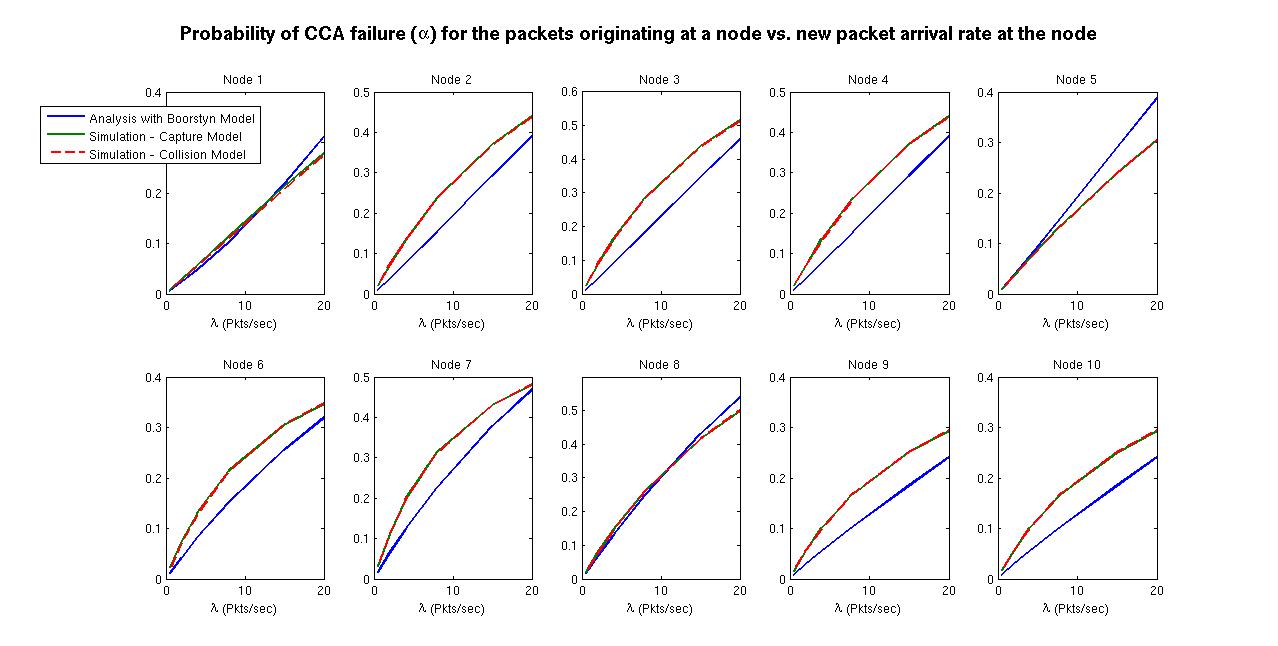}
	\caption{Plots of CCA failure rate, $\alpha$, for the nodes in Figure~\ref{fig:hidden_10nodes}; blue solid line indicates analytical results, green solid line indicates simulation with Capture Model, and Red dotted line indicates simulation with Collision + PER Model}
\label{fig:measure_alpha_for_tree-n10-CS3-PER0.01}
\end{center}
\end{figure}

\begin{figure}[ht]
\begin{center}
	\includegraphics[height=7cm, width=9.5cm]{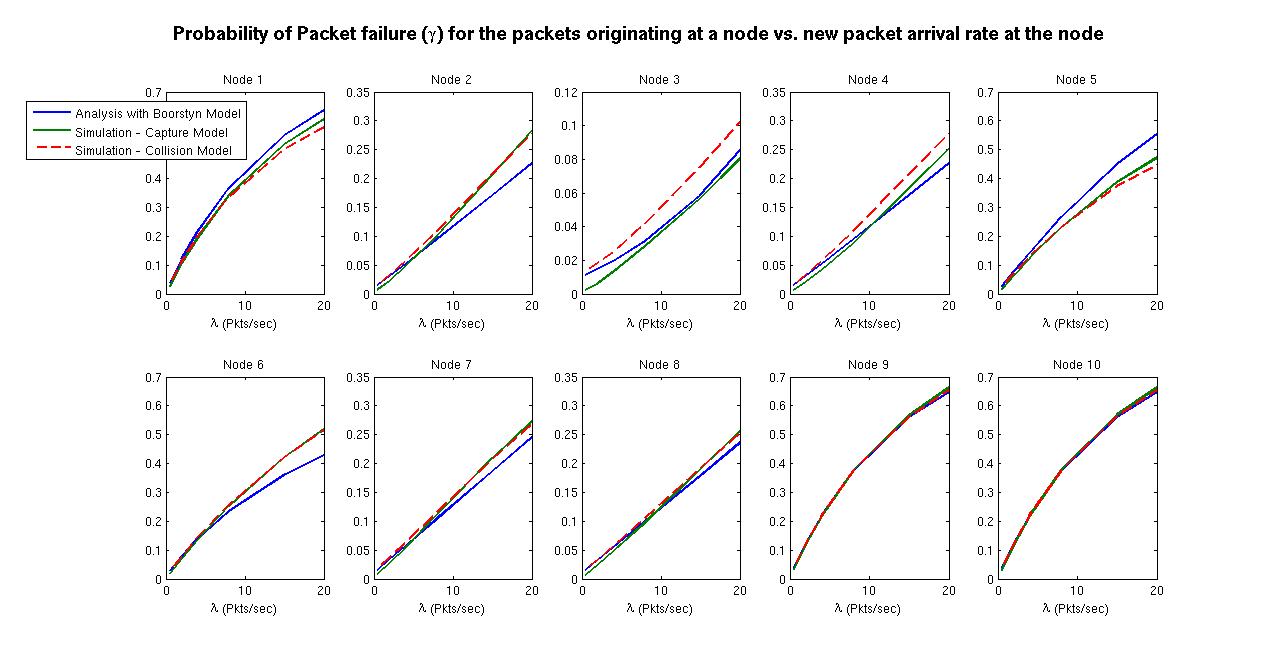}
	\caption{Plots of the packet failure probability, $\gamma$, for the nodes in Figure~\ref{fig:hidden_10nodes}; blue solid line indicates analytical results, green solid line indicates simulation with Capture Model, and Red dotted line indicates simulation with Collision + PER Model}
\label{fig:measure_gamma_for_tree-n10-CS3-PER0.01}
\end{center}
\end{figure}

\begin{figure}[ht]
\begin{center}
	\includegraphics[height=7cm, width=9.5cm]{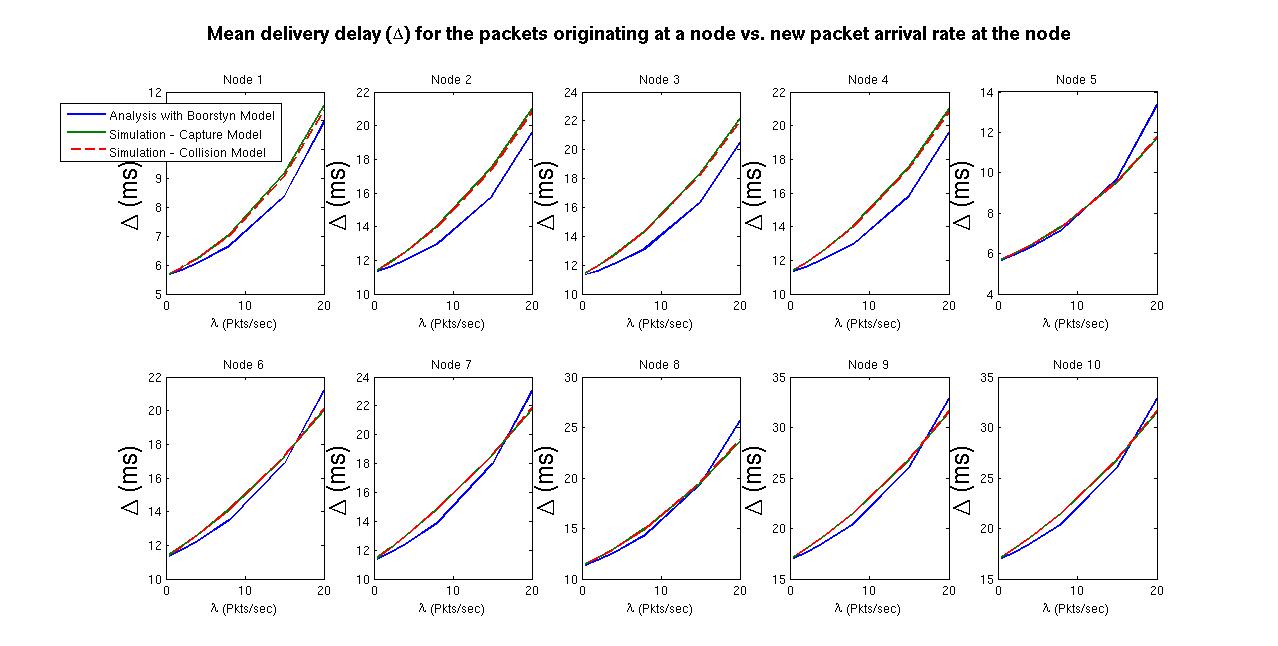}
	\caption{Plots of the mean delivery delay, $\Delta$, for the nodes in Figure~\ref{fig:hidden_10nodes}; blue solid line indicates analytical results, green solid line indicates simulation with Capture Model, and Red dotted line indicates simulation with Collision + PER Model}
\label{fig:measure_delta_for_tree-n10-CS3-PER0.01}
\end{center}
\end{figure}

\begin{figure}[ht]
\begin{center}
	\includegraphics[height=7cm, width=9.5cm]{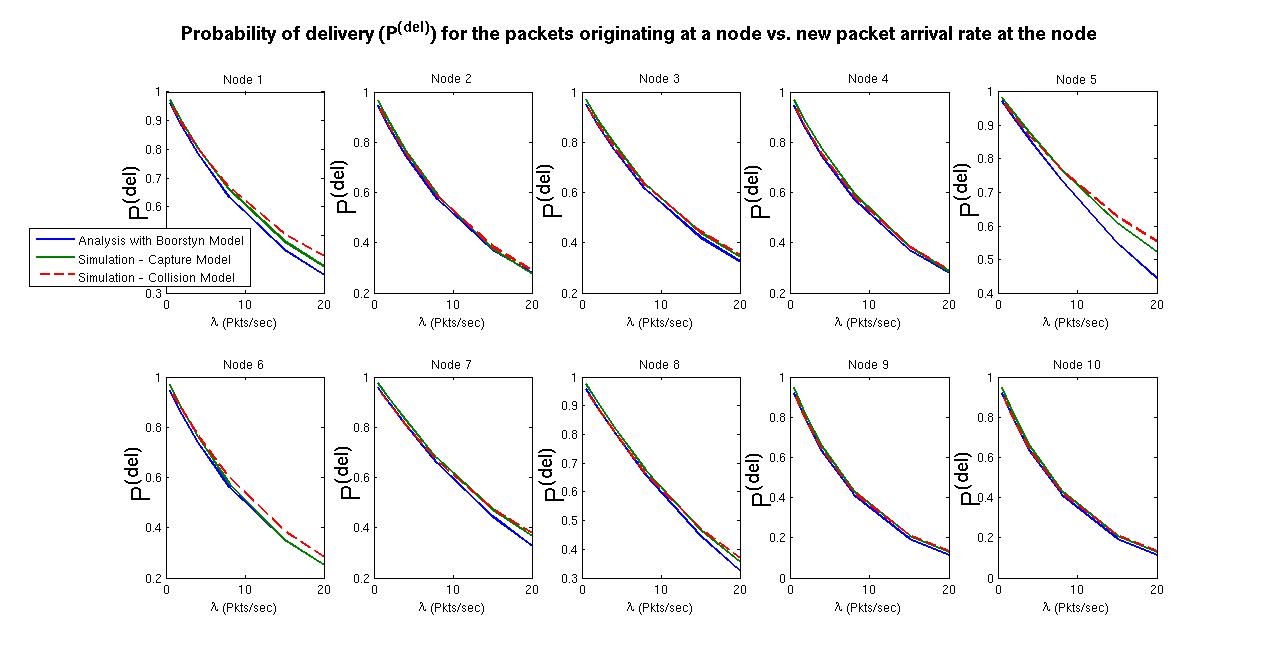}
	\caption{Plots of the average probability of delivery, $P^{(\mathsf{del})}$, for the nodes in Figure~\ref{fig:hidden_10nodes}; blue solid line indicates analytical results, green solid line indicates simulation with Capture Model, and Red dotted line indicates simulation with Collision + PER Model}
\label{fig:measure_pdel_for_tree-n10-CS3-PER0.01}
\end{center}
\end{figure}

\gap
\noindent
\textbf{Observations:}
\begin{enumerate}
\item We observe from the plots that the errors in the measures increase at higher arrival rates, i.e., the accuracy of the analysis decreases at higher arrival rates. For example, whereas for Nodes 1 to 4, and Nodes 6 to 10, the errors in the packet failure probability, $\gamma_i$, were between 1.5\% and 20\%, for Node 5, the error was up to 22.2\% with arrival rate of 20 packets per second. We will see, however, that the practical operating range of these networks is closer to 1 measurement per second from each node, for which value the approximation is excellent.
\item The differences in the simulation measures under the Capture Model and the Collision + Link Error Model were well within 10\% for all nodes at lower arrival rates, and within 20\% even at an arrival rate of 20 packets per second. Hence, we conclude that the Collision + PER Model does not yield performance significantly different from the Capture Model. For other examples illustrating this fact, see \cite{sanjaymeth}. 
\end{enumerate}

\subsection{Extensive Simulation Results}
We now present a summary of an extensive simulation study, where we only use Collision + Link Error as the model. We have performed simulations with a variety of scenarios where one or more of the following features were varied: the topology, number of nodes in the network, average number of nodes in the CS range of a node, and the probability of data packet error on the link due to noise. For compact representation of the various cases, we use the following notation:
\begin{itemize}
     \item[(i)] tree-n$N$-CS$m$-PER$l$: These are the cases in which the network is a tree with $N$ nodes (excluding the base station), all of which are sources. On an average, $m$ nodes are in the CS range of a node, i.e., $m=\frac{\sum_{i=1}^N \text{number of nodes in the CS range of node $i$}}{N}$. Moreover, the PER takes the same value, $l$, on all the links. 
     \item[(ii) ] tree$\mathsf{R}$-n$N$-CS$m$-PER$l$: These are the cases in which the network is a tree with $N$ nodes (excluding the base station) such that 10 nodes among them are source nodes, and the rest of the nodes are relay nodes (the ``R'' in ``treeR'' signifies that there are relays as well as sources). On an average, $m$ nodes are in the carrier sense range of a node, and the PER takes the same value, $l$, on all the links. 
     \item[(iii) ] star-n$N$-CS$m$-PER$l$: These are the cases in which the network is a star with $N$ nodes (excluding the base station), all of which are sources, positioned symmetrically on a circle centred at the base station. The $m$ nodes nearest to a node are in the carrier sense range of it, and the PER on all the links has the same value $l$. 
          \item[(iv) ] line-n$N$-CS$m$-PER$l$: These are the cases in which the network is a line with $N$ nodes (excluding the base station) all of which are sources. $m$ nodes on either side of a node are in the carrier sense range of it. The PER on all the links is the same, and is equal to $l$. 
\end{itemize} 
For example, the case in Figure~\ref{fig:hidden_10nodes} is tree-n10-CS3-PER0.01.

We compute the fractional errors in performance measures as $\frac{\text{Simulation} - \text{Analysis}}{\text{Simulation}}$. The Boorstyn et al. \cite{boorstyn-etal87CSMA-throughput-analysis} model based analysis (see \ref{subsubsec:Ti-eff}) is employed for computing $T_i^{(\mathsf{eff})}$. 

Table~\ref{tab:error for various topologies} tabulates the accuracy of the analysis compared to simulation, and also the computation times of both simulation and analysis. We use the following notation to indicate the range of errors: (i) \checkmark indicates that the error is within $\pm$ $10\%$; (ii) $+$ indicates overestimate by analysis with error from $10\%$ to $25\%$;  (iii) $++$ indicates overestimate by analysis with error more than $25\%$; (iv) $-$ indicates underestimate by analysis with error from $10\%$ to $25\%$; and (v) $--$ indicates underestimate by analysis with error more than $25\%$. The entry in every row is $a, b$, where $a, b \in \{\checkmark, +, ++, -, -- \}$ with $a$ denoting the error summary for a range of $\lambda$ where, the discard probability, $\delta_i \leq 0.01$, and $b$ denoting the error summary for  $\lambda > 2$ pkts/sec. 

\begin{table*}[ht]
\begin{center}
\begin{tabular}{@{}|c|c|c|c|c|c|@{}}
\hline
Sl. No	&	Topology &	$\overline{P}^{(\mathsf{del})}$ & $\overline{\Delta}$ & Simulation Time & Analysis Time\\
        &            &                               &                    & (Min, Mean, Max) in seconds & in seconds \\
        &            &								 &					  & per arrival rate until 10pkts/sec & per arrival rate until 10pkts/sec \\	
\hline
1 & tree$\mathsf{R}$-n23-CS4-PER0.01 &\checkmark, $-$ & \checkmark, \checkmark & $3120, 9720, 17580$	& $16$ \\
\hline 
2 & tree$\mathsf{R}$-n20-CS3-PER0.01 &\checkmark, $-$ &	\checkmark, \checkmark & $2220, 7740, 15060$	&	$12$ \\
\hline
3 & tree$\mathsf{R}$-n20-CS3-PER0.01 & \checkmark, $-$ & \checkmark, \checkmark & $2580, 8640, 16800$	& $14$ \\
\hline
4 & tree$\mathsf{R}$-n19-CS3-PER0.01 & \checkmark, $-$ & \checkmark, \checkmark	& $2460, 7680, 16020$	&  $10$ \\
\hline
5 & tree$\mathsf{R}$-n19-CS2-PER0.01 & \checkmark, \checkmark &	\checkmark, \checkmark	& $1800, 7440, 14340$ & $10$ \\
\hline
6 & tree$\mathsf{R}$-n20-CS3-PER0.01 & \checkmark, $-$ & \checkmark, \checkmark	& $2880, 9480, 17040$	& $16$ \\
\hline
7 & tree$\mathsf{R}$-n20-CS2-PER0.01 &	\checkmark, $--$ &	\checkmark, \checkmark & $2280, 8340, 15180$	& $12$ \\
\hline
8 & tree$\mathsf{R}$-n22-CS4-PER0.01 & \checkmark, $--$	& \checkmark, \checkmark & $3420, 11820, 19680$	& $36$ \\
\hline
9 & tree$\mathsf{R}$-n19-CS2-PER0.01 & \checkmark, $-$	& \checkmark, \checkmark &	$2040, 7440, 14340$	& $14$ \\
\hline
10 & tree$\mathsf{R}$-n19-CS3-PER0.01 &	\checkmark, $-$	& \checkmark, \checkmark &	$2460, 8760, 17040$	& $17$ \\
\hline
11 & tree$\mathsf{R}$-n20-CS3-PER0.01 &	\checkmark, $-$	& \checkmark, \checkmark & $2640, 9180, 17460$	& $17$ \\
\hline
12 & tree$\mathsf{R}$-n19-CS3-PER0.01 & \checkmark, $-$	& \checkmark, \checkmark &	$2340, 7920, 16740$	& $15$ \\
\hline
13 & tree$\mathsf{R}$-n19-CS3-PER0.01 &	\checkmark, $-$	& \checkmark, \checkmark &	$2400, 8280, 16980$	& $12$ \\
\hline
14 & tree$\mathsf{R}$-n20-CS3-PER0.01 &	\checkmark, $-$	& \checkmark, \checkmark & $2340, 8340, 16440$	& $11$ \\
\hline
15 & tree$\mathsf{R}$-n19-CS3-PER0.01 &	\checkmark, $-$	& \checkmark, \checkmark &	$2280, 7860, 16260$	& $11$ \\
\hline
16 & tree$\mathsf{R}$-n19-CS2-PER0.01 & \checkmark, $--$ &	\checkmark, \checkmark & $1740, 7260, 14220$ & $9$ \\
\hline
17 & tree$\mathsf{R}$-n20-CS3-PER0.01 &	\checkmark, $-$	& \checkmark, \checkmark &	$2520, 8520, 16020$	& $11$ \\
\hline
18 & tree$\mathsf{R}$-n19-CS2-PER0.01	& \checkmark, $-$ &	\checkmark, \checkmark	& $2040, 7560, 14760$	& $10$ \\
\hline
19 & tree$\mathsf{R}$-n19-CS2-PER0.01 &	\checkmark, $-$	& \checkmark, \checkmark &	$1920, 7140, 14460$ & $9$ \\
\hline
20 & line-n10-CS2-PER0.01 & \checkmark, \checkmark & \checkmark, \checkmark & 780, 6240, 12060& 5 \\
\hline
21 & line-n10-CS3-PER0.01 & \checkmark, \checkmark & \checkmark, \checkmark & 1080, 7860, 13980& 8 \\
\hline
22 & line-n10-CS4-PER0.01 & \checkmark, \checkmark & \checkmark, $+$ & 1260, 8700, 16080 & 12 \\
\hline
23 & star-n20-CS9-PER0.01 & \checkmark, \checkmark & \checkmark, \checkmark &900, 6480 ,12660& 36 \\
\hline
24 & star-n20-CS11-PER0.01 & \checkmark, \checkmark & \checkmark, \checkmark &1140, 7620, 14880 & 48 \\
\hline
25 & tree-n13-CS2-PER0.01 & \checkmark, \checkmark & \checkmark, \checkmark & 960, 7320, 13260 & 5 \\
\hline
\end{tabular}
\caption{Error in measures $P^{(\mathsf{del})}_i$ and $\Delta_i$ where $\overline{P}^{(\mathsf{del})}$ and $\overline{\Delta}$ are the errors averaged over all nodes for every input arrival rate. Note that run timing of simulation and analysis are in \emph{seconds}.} 
\label{tab:error for various topologies}
\end{center}
\end{table*}

\gap
\noindent
\textbf{Observations and Discussion:}
\begin{enumerate}
\item As can be observed from Table~\ref{tab:error for various topologies}, for small arrival rates at which the discard probability, $\delta_i \leq 0.01$, the errors in both $\overline{P}^{(\mathsf{del})}$, and $\overline{\Delta}$ are less than 10\% in all the scenarios tested, whereas at higher arrival rates, the error in $\overline{P}^{(\mathsf{del})}$ sometimes exceeds 25\% (scenarios 7,8, and 16), which is in agreement with our earlier observation from Figure~\ref{fig:measure_pdel_for_tree-n10-CS3-PER0.01} as well. However, the error in $\overline{\Delta}$ even at higher arrival rates was within 10\% for all but one scenario (scenario 22), where the analysis overestimated the delay within an error of 25\%. 

One possible explanation for the degradation in the accuracy of the analysis at higher arrival rates can come from the discussion in Section~\ref{sec:dtmc-stability}. Note that at higher arrival rates, the queue non-empty probability of the nodes are typically higher, and hence the condition $\sum_{i=1}^N q_i < 1$ (see Theorem~\ref{thm:stability}) is more likely to be violated, in which case, the system may not be stable, and the validity of the fixed point analysis is questionable. For example, note from Figure~\ref{fig:measure_sum_qi} that for the example topology in Figure~\ref{fig:hidden_10nodes}, $\sum_{i=1}^N q_i$ approaches 1 at around $\lambda = 10$ pkts/sec, and it can be seen from Figures~\ref{fig:measure_alpha_for_tree-n10-CS3-PER0.01}, \ref{fig:measure_delta_for_tree-n10-CS3-PER0.01}, \ref{fig:measure_gamma_for_tree-n10-CS3-PER0.01}, and \ref{fig:measure_pdel_for_tree-n10-CS3-PER0.01} that the errors in the measures are more pronounced for $\lambda \geq 10$ pkts/sec.

\item We also observe from Table~\ref{tab:error for various topologies} that the average time to analytically compute the performance measures for each arrival rate (using the Boorstyn~et.al \cite{boorstyn-etal87CSMA-throughput-analysis} model for computing $T^{(\mathsf{eff})}_i$) was of the order of seconds, whereas the average simulation time ran into several hours. 

Further, it was observed that for hidden-nodes networks with 20 nodes, and an average of 3 nodes in the CS range of a node, the analysis with Boorstyn~et.al \cite{boorstyn-etal87CSMA-throughput-analysis} model for computing $T^{(\mathsf{eff})}_i$ takes around $13$ secs per arrival rate, and if $T^{(\mathsf{eff})}_i$ is computed using $M/D/\infty$, then the analysis takes around $9$ secs per arrival rate. Hence we can trade-off between the accuracy of analysis and computation time since analysis with $M/D/\infty$ model usually incurs more error ($> 25\%$) at higher arrival rates, but less than $10\%$ for $\lambda < 0.5$ pkts/sec. 

\item Finally, it was observed that the packet discard probability at a node increases to an impractical value much before the average delay at the node becomes substantial, e.g., tree-n10-CS3-PER0.01-Node-1 discards $2$-$5$ packets for every $1000$ packets when the average node delay is $5$-$6$ msec. Hence, the overall packet delivery probabilities for the sources act as performance bottlenecks for these networks.  
\end{enumerate}


\section{Network Design} \label{sec:network-design}
In this section, we shall consider a problem of QoS constrained network design for a given positive traffic arrival rate to demonstrate the usefulness of the analytical model in an iterative network design process. 

Consider the following setting: \emph{a set of sensor nodes and a base station (BS) are deployed over an area; each sensor node is a data source, and can also act as a relay. The transmit power level of each node can be adjusted over a range.} Consider the complete graph over these nodes. Each directed link in this graph will have a certain packet error rate (PER) that \emph{depends on the transmit power level of the sender node on that link.} 

We consider the following network design problem: for a given traffic arrival rate $\lambda > 0$ at the sensor nodes, \emph{minimize the maximum transmit power level used by the sensor nodes, such that the resulting network has the following properties.}

\begin{enumerate}
\item Each sensor node has a path to the BS, with the PER on each link in each path being upper bounded by a predefined target threshold $p$.
\item For the given $\lambda$, the end-to-end packet delivery probability (i.e., the probability that a packet is not discarded) on any path is at least $p_{\mathrm{del}}$.
\item For the given $\lambda$, the mean delay (computed over the successfully delivered packets) on any path is upper bounded by a predefined target $d_{\max}$.
\end{enumerate}

Note that for a design (network) to satisfy the QoS objectives for a given arrival rate, it is \emph{necessary} that the network satisfies the QoS objectives under zero/light traffic load, i.e., in the limit as $\lambda \to 0$, in which limit a packet that enters the network departs from the network before another packet arrives, i.e., each packet traverses the network $\emph{alone}$, prompting us to call this limit the $\emph{lone packet}$ traffic model (for a more formal proof\footnote{A formal proof is necessary for this seemingly obvious statement since, in CSMA/CA networks, in general, the performance is not monotone with the arrival rates (see, e.g., \cite{borst}); hence, the statement about the lone-packet traffic model needs to be made with care.} of this fact, see \cite{fullpaper}). We, therefore, adopt the following two step approach:

\begin{enumerate}
\item We first focus on the QoS constrained network design problem under the lone-packet traffic model in Section~\ref{subsec:lone-packet-design}; we formulate the problem as a network design problem on graphs, and propose an algorithm to solve the problem optimally.
\item Then, in Section~\ref{subsec:cont-traffic-design}, we combine the lone-packet based design algorithm with the analytical tool developed in Section~\ref{sec:analytical-model} to address the more general (and more complex) problem of QoS constrained network design for a given positive traffic arrival rate.
\end{enumerate}
  
Note that we can ensure a given target PER on a link by ensuring that the \emph{average} received power across the link (averaged over shadowing and fading) meets a received power target. Such a target average received power would
be obtained by deriving a margin above the minimum value of received power, from statistics of shadowing and fading. If we take this worst case approach over the joint pdf of shadowing and fading, then for a given target PER, the required transmit power on a link is a non-decreasing function of the link length. Hence, the problem of minimizing the maximum transmit power is \emph{equivalent to minimizing the maximum link length in the resulting network}.  

\subsection{Network Design under the Lone-Packet Model} 
\label{subsec:lone-packet-design}

Given the target PER $p$ on any link, the mean delay on a link under the ``lone packet'' model can be computed using an elementary analysis (see \cite{abhijitmeth}), taking into account the backoff behavior of 802.15.4 CSMA/CA, and using the backoff parameters given in the standard \cite{IEEE802-15-4-06std}. Then, to meet the mean delay requirement of $d_{\max}$ on a path with $h$ hops, we require that

\begin{equation}
h\leq \left\lfloor\frac{d_{\max}}{\overline{d}_{single-hop}}\right\rfloor\triangleq h_{\max}^{delay}
\end{equation}
where, $\overline{d}_{single-hop}$ is the mean link delay computed as explained earlier.

Again, given the target link PER of $p$, and the number of retries $r$ (obtained from the standard) before a packet is discarded on a link, the packet discard probability on a link can be obtained as $q\:=\:p^{r+1}$. Hence, to ensure a packet delivery probability of at least $p_{del}$ on a path with $h$ hops, we require that

\begin{eqnarray}
h &\leq \frac{\ln{p_{del}}}{\ln{(1-q)}}&\triangleq h_{\max}^{delivery}
\end{eqnarray}

Hence, to ensure the QoS constraints under the ``lone packet'' model, we require that the \emph{hop count on each path is upper bounded by $h_{\max}\:=\:\min\{h_{\max}^{delay},h_{\max}^{delivery}\}$}.

Thus, the problem of QoS constrained network design under the ``lone packet'' model can be reformulated as the following graph design problem: 

\emph{Given a set of source nodes $Q$, and a base station (BS), indexed as node 0, consider the graph $G=(V,E)$, where $V=Q\cup \{0\}$, and the edge set $E$ consists of all feasible edges, i.e., edges with PER $\leq\: p$ at maximum possible transmit power. Given a hop count constraint $h_{\max}$, the problem is to extract from $G$, a spanning tree on $Q$ rooted at the BS, such that the hop count on any path is upper bounded by $h_{\max}$, and the maximum edge length in the spanning tree is minimized}.

We call this, the MinMax Spanning Tree with Hop Constraint (MMST-HC) problem. 

The following is an algorithm to obtain an optimal solution to the MMST-HC problem. At each iteration, we prune all edges with edge length more than or same as the maximum edge length in the current feasible solution, form a shortest path tree (SPT) with hop count as cost using only the remaining edges, and keep doing this until the resulting SPT violates the hop constraint, at which point we stop and declare the last feasible solution as the final solution.

\subsubsection{\textbf{SPTiEP: Shortest Path Tree based iterative Edge Pruning Algorithm}}

\begin{itemize}
\item[(i)] \textbf{Initialize: }Set $k\leftarrow 0$, $G^{(0)}\leftarrow G$
\item[(ii)] \textbf{Checking feasibility: }In iteration $k$, find a shortest path tree (SPT) $T^{(k)}$ on the graph $G^{(k)}$. Check if all the paths from the sources to the BS in $T^{(k)}$ satisfy the hop constraint.

\begin{itemize}
\item If the hop constraint is not met for some of the sources in $T^{(k)}$, \textbf{STOP}. 
\begin{itemize}
\item If $k=0$, declare the problem \emph{infeasible}. 
\item If $k>0$, output $T^{(k-1)}$ as the final solution.
\end{itemize}

\item If the hop constraint is met for all the sources in $T^{(k)}$, proceed to the next step.
\end{itemize}
\item[(iii)] \textbf{Edge pruning: }Let $\overline{w}^{(k)}$ be the maximum edge length in $T^{(k)}$. Remove from $G^{(k)}$, all edges of length $\geq \overline{w}^{(k)}$ to obtain the graph $G^{(k+1)}$.
\item[(iv)] \textbf{Iterate: }Set $k\leftarrow k+1$. Go to Step 2.
\end{itemize}

Since the time complexity of finding an SPT in a graph with $N$ nodes is $O(N\log N)$, it can be easily verified that the time complexity of the SPTiEP algorithm is $O(N^3\log N)$ for a graph with $N$ nodes. Thus, \emph{the SPTiEP algorithm is polynomial time}.

\subsubsection{\textbf{Proof of correctness of SPTiEP}}

We define
\begin{description}
\item $\mathcal{F}^{(k)}\:=\:\{T\subset G^{(k)}: \text{$T$ satisfies the hop constraint}\}$: set of all feasible solutions contained in the graph $G^{(k)}$
\item $w_{opt}$: The minmax edge length, i.e., the maximum edge length in an optimal solution
\end{description}

Also recall the definitions of $T^{(k)}$ and $\overline{w}^{(k)}$ from the description of the algorithm.

Clearly, $\mathcal{F}^{(k)}$ is non-empty \emph{if and only if} $T^{(k)}\in \mathcal{F}^{(k)}$.

To prove the correctness of SPTiEP algorithm, it is enough to show the following:

\begin{proposition}
\label{prop:sptiep-correctness}
Given that $\mathcal{F}^{(k)}\:\neq \emptyset$,  
\begin{equation*}
\mathcal{F}^{(k+1)}\:=\emptyset\: \Rightarrow w_{opt}=\overline{w}^{(k)}
\end{equation*}
\end{proposition}

\begin{proof}
See the Appendix.
\end{proof}

\subsection{Network Design for a Given Positive Arrival Rate} 
\label{subsec:cont-traffic-design}

Now we come back to the more general problem of QoS constrained network design for a given positive arrival rate. This problem is rendered much more difficult compared to its lone-packet version due to the complex stochastic interaction between contending nodes, which, unlike the lone-packet model, makes it hard to map the QoS constraints to simple explicit constraints on certain graph properties. Therefore, unlike the lone-packet version where we posed the problem as a pure graph design problem, in the positive traffic case, we shall use the analytical model explicitly in an iterative design process to evaluate for QoS, the designs obtained systematically in every iteration. While traditional network simulation tools can also be used, in principle, to evaluate a given network for QoS, \emph{the time required for such network simulation is significantly more than that required by the analysis} (as we saw in Table~\ref{tab:error for various topologies}, and shall see again in our numerical experiments in Section~\ref{subsec:design-exp}), and that makes network simulation, an impractical option in an iterative design process.   

Also, a naive approach to the design problem would be to consider all possible trees from the given graph $G$, and evaluate each of them for QoS (using either the analysis, or network simulation), and choose one whose maximum edge length is minimum among all those that meet the QoS objectives. While this exhaustive search appears to be the only way that is guaranteed to obtain a feasible solution whenever there exists one, this approach has exponential time complexity (since the number of possible trees is exponential in the number of nodes), and is therefore, not practical. 

We shall present below, a polynomial time algorithm for the proposed positive traffic design problem, using the SPTiEP algorithm along with the analytical model; note that because of the stochastic nature of the interaction between contending nodes in different possible networks, in absence of an exhaustive evaluation of all possibilities, the algorithm is not theoretically guaranteed to return a feasible solution whenever there exists one; however, in our numerical experiments, the algorithm was always found to return a feasible solution, and moreover, simulations confirmed that designs proposed with the analytical tool did meet the QoS requirements (see Section~\ref{subsec:design-exp}).   

\subsubsection{\textbf{Extended SPTiEP: An algorithm for QoS constrained network design at given positive arrival rate $\lambda > 0$}} 
As we had mentioned before, to meet the QoS objectives at a positive arrival rate, it is \emph{necessary} (but not sufficient) to meet the objectives under the lone-packet model. Hence, our design still needs to satisfy the hop count constraint $h_{\max}$ derived from the QoS constraints under the lone-packet model. With this in mind, we proceed as follows: if the outcome of the SPTiEP algorithm meets the QoS constraints at the given $\lambda$, then that is an optimal solution (since it is an optimal solution for the lone-packet design problem, and it is also QoS feasible for the given $\lambda>0$). If, however, the SPTiEP outcome does not meet the QoS constraint, we need to change the design by \emph{adding some of the edges of greater length that were pruned in course of the SPTiEP algorithm}; note that adding edges is the only option as pruning any more edge from the SPTiEP solution will cause us to violate the hop constraint, and hence the lone-packet QoS. The detailed steps are presented below.

\begin{enumerate}
\item \textbf{Lone-packet design: }Run the SPTiEP algorithm on the graph $G$ to obtain a tree $T_0$.
\begin{itemize}
\item If $T_0$ does not satisfy the hop constraint, declare the problem infeasible, as we cannot satisfy the QoS objectives even for $\lambda = 0$.
\item Else, go to the next step.
\end{itemize}
\item Set $k\leftarrow 0$. Mark all edge lengths in $G$ as \emph{not examined}. 
\item \textbf{Checking feasibility for $\lambda>0$: }Evaluate $T_k$ for QoS requirements (i.e., $p_{del}$ and $d_{\max}$) at the given arrival rate $\lambda$, using the analytical model.\footnote{for the purposes of this design, we completely trust the outcome of the analytical model, i.e., we assume the analytical model to be 100\% accurate.}
\begin{itemize}
\item \textbf{Stopping criteria 1: }If QoS is met, output $T_k$ as the final solution.
\item \textbf{Stopping criteria 2: }If QoS is not met, and all edge lengths in $G$ have been \emph{examined}, declare the problem \emph{possibly} infeasible.
\item Else, go to next step.
\end{itemize}
\item Let $\overline{w}_k$ be the maximum edge length in $T_k$. Identify the least edge length $> \overline{w}_k$ in the graph $G$ that is not yet \emph{examined}; let us denote this as $w_{least,k}$. 

\item \textbf{Edge augmentation: }Augment $T_k$ with all edges (in $G$) of length $\leq w_{least,k}$, to obtain the graph $G_{k+1}$. Mark all the edge lengths in $G_{k+1}$ as \emph{examined}.

\gap
\noindent
\comment Since our objective is to minimize the maximum edge length, in Steps 4 and 5, we add back the edges pruned during SPTiEP \emph{in increasing order of their lengths}.

\item \textbf{Redesign: }Find an SPT $T_{k+1}$ in $G_{k+1}$. Observe that if $T_k$ satisfies the hop constraint, $T_{k+1}$ also satisfies the hop constraint.
\item \textbf{Iterate: }Set $k\leftarrow k+1$. Go to Step 3.

\gap
\noindent
\comment In Steps 6 and 7, we check if the resulting shortest path tree in the augmented graph satisfies the QoS constraints for the given positive arrival rate. 

Note that in Step 6, there could be several possible SPTs on the graph $G_{k+1}$; while the total network wide traffic load, as well as the total number of nodes is the same in all of these SPTs, the individual loads on the various nodes may vary from one SPT to another, resulting in potentially different delay and delivery performance. Since it is not possible to evaluate all the possible SPTs in polynomial time, we find and evaluate only one of them, and move on to a higher power/edge length design in case of QoS not being satisfied; this may lead to suboptimal design. Moreover, for the same reason, the algorithm may not always return a feasible solution even when there exists one.  
\end{enumerate}

\gap
\noindent
\textbf{Remarks:}
\begin{enumerate}
\item It is interesting to note that the above heuristic is basically the reverse procedure of the SPTiEP algorithm, with the additional step of evaluating each design using the analytical model. At the end of each iteration, the maximum edge length in the solution changes at most to the immediate higher value not used in previous iterations. We stop as soon as a QoS feasible solution is obtained in some iteration. 
\item An alternate approach could be to start with a shortest path tree on the entire graph, and then prune edges, checking for feasibility with positive load in each iteration. But since meeting the lone-packet QoS is necessary to meet the positive-load QoS, we have adopted the above approach of first designing a network satisfying the lone-packet QoS, and then backtracking.      
\end{enumerate}
 
\subsection{Numerical Experiments}
\label{subsec:design-exp}
\begin{table*}[Ht]
  \centering
\caption{Execution time comparison of the analysis based design algorithm and qualnet simulation}
\label{tbl:exec_time_compare}
\footnotesize
  \begin{tabular}{|c|c|c|c|c|c|c|c|c|c|c|c|c|}\hline
    Scenarios & \multicolumn{6}{c|}{Algorithm execution time} & \multicolumn{6}{c|}{Qualnet simulation time}\\
     & \multicolumn{6}{c|}{in sec} & \multicolumn{6}{c|}{in sec}\\
     & \multicolumn{3}{c|}{SPTiEP} & \multicolumn{3}{c|}{E-SPTiEP} & \multicolumn{3}{c|}{SPTiEP outcome} & \multicolumn{3}{c|}{E-SPTiEP outcome}\\
 &  Average & Max & Min & Average & Max & Min & Average & Max & Min & Average & Max & Min\\
 \hline
   30 & 0.0468 & 0.2532 & 0.0300 & 0.3827 & 0.5084 & 0.1609 & 2412.3 & 3000.6 & 2041.2 & 1765.7 & 2233.8 & 1326.6\\
  \hline
\end{tabular}
\normalsize
\end{table*}

To demonstrate the strengths and limitations of the analysis in an iterative design process, we performed several test runs of the proposed algorithm on randomly generated network scenarios. 

30 random networks were generated in a $50\times 50\:m^2$ area as follows: The entire area was partitioned into square cells of side 10 meters. Consider the lattice created by the corner points of the cells. 10 source nodes were placed at random over these lattice points. The minimum power level of the nodes is assumed to be sufficient so that all nodes are within CS range of one another, i.e., we have a ``no hidden node'' scenario. We chose the target link PER to be 0.01, $p_{del}=95\%$, and $d_{\max}=25\:msec$. These result in a hop constraint of $h_{\max}=5$ for the ``lone packet'' model. Also, the traffic arrival rate at each source was chosen to be $\lambda = 1$ packet/s, which is actually quite adequate for many wireless sensing applications.

The design algorithm using the analytical model was run on the 30 random scenarios. \emph{In each case, the algorithm was found to return a feasible solution.} 

Also, to validate the solution provided by the algorithm, we performed Qualnet simulation on both the outcome of the SPTiEP algorithm as well as the final solution for each of the 30 cases. The observations are as follows:

\textbf{Observations:}
\begin{enumerate}
\item In all 30 cases, simulations of the final solution showed QoS performance better than the target. That is, \emph{the solution provided by the algorithm using the analytical model did actually satisfy the QoS constraints as verified by the simulation in all cases}. 
\item In 12 out of the 30 cases, simulations suggested that the SPTiEP solution (``lone packet'' design) met the QoS constraints even for the given $\lambda = 1$ packet/s (and hence was optimal), but the analytical model suggested otherwise. Thus, \emph{the analytical model was somewhat conservative in its prediction of QoS}, which eventually led to a design that used more power than the optimal. 
\item However, while the analytical model was conservative in its prediction of QoS, the final design met the QoS objectives and predicted the performance to within 10\%., i.e., the model was \emph{quite accurate} in its prediction. 
\item Finally, we compared the time taken to run the design algorithm using the analytical model against the time taken to run the Qualnet simulations on the designed networks in the 30 cases. To save time, \emph{in each case, we performed Qualnet simulations only for the lone packet design, and the final design obtained using E-SPTiEP algorithm}. While the algorithm was run in MATLAB 7.11 on a Windows Vista based Dell Inspiron 1525 laptop with 3 GB RAM, and 2 GHz Intel Core 2 Duo CPU, the simulations were run in Qualnet 4.5 on a Linux based Dell server with 32 GB main memory, and 3 Ghz clock speed. The results are summarized in Table~\ref{tbl:exec_time_compare}.

From Table~\ref{tbl:exec_time_compare}, we see that while \emph{the execution time of the analysis based algorithm is of the order of milliseconds}, simulation of each of the designed networks takes several minutes, which indicates that \emph{if the QoS evaluation step in the iterative design is performed using simulation instead of the analysis, the execution time of the algorithm can go well beyond an hour}. 
\end{enumerate}

From the above discussion, we can summarize the following strengths and limitation of the analytical model for use in a network design process.

\begin{enumerate}
\item \textbf{Strengths:}
\begin{itemize}
\item Design based on analysis is much \emph{faster} compared to that based on network simulation, and the predictions are \emph{accurate} to well within 10\% compared to the simulation results. Such speed and accuracy also make the analytical tool a good choice for \emph{online/field-interactive} design process, or even for as-you-go deployment of an impromptu wireless network.
\item Designs provided by the analysis based algorithm do actually work, i.e., satisfy the QoS constraints in practice, as validated by simulations of the designed networks. 
\end{itemize}
\item \textbf{Limitation: }The analysis, while quite accurate compared to the simulation, is somewhat \emph{conservative in its prediction of QoS performance}. This may sometimes lead to \emph{suboptimal} design, or even \emph{declaration of infeasibility}, when actually there may exist a feasible solution. 
\end{enumerate}


\section{Conclusion} \label{sec:conclusion}

We have developed an approximate stochastic model for the performance analysis of beacon-less \mbox{IEEE 802.15.4} multihop wireless networks, with arrivals. The model permits the estimation of several time-average performance measures. Our model is accurate at small arrival rates (at which packet discard probability is small) in terms of the packet discard probability, failure probability and throughput, and delay. We calculated the mean end-to-end delays and packet delivery probabilities for each source in the network for specific packet generation rates at the source nodes. The results suggest that, for the relatively small size tree networks that we have studied, to operate in the low-discard low-delay region, the packet arrival rates at nodes should not be greater than a packet every few seconds (e.g. a packet inter-generation time of 5 to 10 seconds at the sources). Finally, we have formulated and solved a problem of QoS constrained network design for given positive traffic arrival rate to demonstrate the usefulness of the analytical model in an iterative network design process.

\section*{Appendix}
\subsection{Proof of Proposition~\ref{prop:irreducible}}
\begin{proof}
First observe that transition from any state to the all-zero state is possible if we consider a series of "only down" transitions, i.e., no further external arrival occurs into any of the queues, and the packets leave the queues either due to successful transmission, or drop due to excessive retries or
CCA failures. \emph{No external arrival} would mean that the number of packets
in the system can only decrease in each transition, and eventually we will
reach the all-zero state. Note that this is possible since the maximum number
of retransmission attempts, and maximum allowed CCA failures are \emph{finite}.

Now, suppose the class is open. Then, there exists a state $i\in \mathcal{C}_{0}$, and a state $j\notin \mathcal{C}_{0}$ such that $i\rightarrow j$, but $j\nrightarrow i$. But from our earlier argument, we know, state $j$ can reach the all-zero state. Since both state $i$ and the all-zero state belong to the class $\mathcal{C}_{0}$, the all-zero state can reach state $i$. It follows that $j\rightarrow i$, which is a contradiction. Hence, the class  $\mathcal{C}_{0}$ is closed. 

Clearly, $\mathcal{C}_{0}$ is aperiodic since starting in the all-zero state, the system can remain in the all-zero state if no external arrivals occur in a slot (since the arrival process is Poisson, this event has a positive probability). 
\end{proof}

\subsection{Proof of Theorem~\ref{thm:stability}}
\begin{proof}
For any $n>0$,
\begin{align}
\prob[\cap_{i=1}^N\{X_i(n)=0\}|\X(0)=\0,\Y(0)=\0]\nonumber\\
= 1\:-\:\prob[\cup_{i=1}^N\{X_i(n) > 0\}|\X(0)=\0,\Y(0)=\0]\nonumber\\
\geq 1\:-\:\sum_{i=1}^N \prob[X_i(n) > 0|\X(0)=\0,\Y(0)=\0]\label{eqn:union-bnd}
\end{align}
where, in writing \eqref{eqn:union-bnd}, we have used the union bound.

Taking limit as $n\to \infty$ on both sides,

\begin{eqnarray}
\lim_{n\to \infty}\prob[\cap_{i=1}^N\{X_i(n)=0\}|\X(0)=\0,\Y(0)=\0]\nonumber\\
\geq 1\:-\:\sum_{i=1}^N q_i\nonumber\\
 > 0,\:\text{when $\sum_{i=1}^N q_i\:<\:1$}\nonumber
\end{eqnarray}
Thus, if $\sum_{i=1}^N q_i<1$, we have that $\lim_{n\to \infty}\prob[\cap_{i=1}^N\{X_i(n)=0\}|\X(0)=\0,\Y(0)=\0]>0$, which, in conjunction with Proposition~\ref{prop:irreducible}, implies that $(\X(n),\Y(n))$ is positive recurrent.
\end{proof} 

\subsection{Proof of Proposition~\ref{prop:sptiep-correctness}}
\begin{proof}
Since $\mathcal{F}^{(k)}\:\neq \emptyset$, 
\begin{equation}
\label{eqn:leq}
w_{opt}\leq\overline{w}^{(k)}
\end{equation}
This is because $\overline{w}^{(k)}$ is the maximum edge length of $T^{(k)}$, a feasible solution, while $w_{opt}$ is the maximum edge length of an optimal solution. 

Again, $\mathcal{F}^{(k+1)}\:=\emptyset\:\Rightarrow\:\nexists\: T\subset G^{(k+1)}$ such that $T$ satisfies hop constraint. 

But by construction (Step 3 of SPTiEP algorithm), $G^{(k+1)}$ contains in it, all trees with maximum edge length $<\:\overline{w}^{(k)}$. Therefore, no tree with maximum edge length $<\:\overline{w}^{(k)}$ is a feasible solution to the MMST-HC problem. Hence, 
\begin{equation}
\label{eqn:notless}
w_{opt}\:\nless \overline{w}^{(k)}
\end{equation}
Combining \eqref{eqn:leq} and \eqref{eqn:notless}, it follows that $w_{opt}=\overline{w}^{(k)}$. 
\end{proof}

\section*{Acknowledgements}
This work was supported by a research grant from the Department of Electronics and Information Technology (DeitY), Government of India, through the Automation Systems Technology (ASTEC) program, and by a DeitY-NSF funded Indo-US project on Wireless Sensor Networks for Protecting Wildlife and Humans.

\bibliography{journal_paper}
\end{document}